\pdfoutput=1

\documentclass[acmsmall]{acmart}

\usepackage{verbatim}
\usepackage{proof}
\usepackage{latexsym}
\usepackage{amssymb}
\usepackage{url}
\usepackage{hyperref}
\usepackage{amsmath}
\usepackage{amsthm}
\usepackage{prooftree}
\usepackage{color}
\usepackage{spi}

\renewcommand{\vDash}{\models}

\newcommand{\OM}{\ensuremath{\mathcal{O\!M}}}

\newcommand{\ttt}{\mathtt{t\hspace*{-.25em}t}}
\newcommand{\fff}{\mathtt{f\hspace*{-.25em}f}}

\newcommand{\boxm}[1]{\mathopen{\big[ #1 \big]}} 
\newcommand{\diam}[1]{\mathopen{\big\langle #1 \big\rangle}}
\newcommand{\respecting}{\mathrel{\textrm{respecting}}}
\newcommand{\match}[1]{\mathopen{\left[ #1 \right]}}
\newcommand{\yields}{\supset}

\newcommand{\var}{\hat u}

\usepackage[all]{xy}

\usepackage{lipsum}                     
\usepackage{xargs}                      
\usepackage[colorinlistoftodos,prependcaption,textsize=small]{todonotes}
\newcommandx{\alwen}[2][1=]{\todo[linecolor=red,backgroundcolor=red!25,bordercolor=red,#1]{Alwen: #2}}
\newcommandx{\ross}[2][1=]{\todo[linecolor=blue,backgroundcolor=blue!25,bordercolor=blue,#1]{Ross: #2}}

\usepackage{spi}

\title[
A Bisimilarity Congruence
for the Applied $\pi$-Calculus
]{
A Bisimilarity Congruence for the Applied $\pi$-Calculus Sufficiently Coarse to Verify Privacy Properties
}

\setcopyright{none}

\author{Ross Horne}

\affiliation{
Computer Science Research Unit, University of Luxembourg 
\\
\texttt{ross.horne@uni.lu}
}

\renewcommand{\shortauthors}{Horne}



\begin{document}

\begin{abstract}
This paper is the first thorough investigation into the coarsest notion of bisimilarity for the applied $\pi$-calculus that is a congruence relation: \textit{open barbed bisimilarity}.
An open variant of labelled bisimilarity (\textit{quasi-open bisimilarity}), better suited to constructing bisimulations, is proven to coincide with open barbed bisimilarity.
These bisimilary congruences are shown to be characterised by an \textit{intuitionistic modal logic} that can be used, for example, to describe an attack on privacy whenever a privacy property is violated.
Open barbed bisimilarity provides a compositional approach to verifying cryptographic protocols, since properties proven can be reused in any context, including under input prefix. Furthermore, \textit{open barbed bisimilarity} is sufficiently coarse for reasoning about security and privacy properties of cryptographic protocols; in constrast to the finer bisimilarity congruence, \textit{open bisimilarity}, which cannot verify certain privacy properties.


\end{abstract}

\maketitle


\section{Introduction}

There has been much debate surrounding bisimilarity in the context of the applied $\pi$-calculus, since the calculus was first introduced~\cite{abadi01popl} as a generalisation and extension of the $\pi$-calculus~\cite{Milner92pi} for verifying cryptographic protocols.
A central concern revolves around the treatment of mobility of channels in the original presentation of bisimilarity. 

According to the original definition, the following two processes
are \textbf{mistakenly} bisimilar.
\[
P \triangleq \nu z. \cout{x}{z,y}.z(w)
\qquad\mbox{v.s.}\qquad
Q \triangleq \nu z. \cout{x}{z,y}
\]
However, these two processes should be distinguished since another process can: receive the pair $\pair{z}{y}$, take the first projection to obtain private channel $z$, and then use that private channel to send a message. Thus there is a distinguishing context\footnote{Here, action $\cout{\fst{u}}{\textit{y}}$ can be read as ``send message $y$ on the channel obtained by taking the first projection of $u$.'' } 
 $\context{\ \cdot\ } \triangleq {x(u).\cout{\fst{u}}{\textit{y}}} \cpar \left\{\ \cdot\ \right\}$,
such that $\context{P}$ can perform two communications to reach a state with no actions, which cannot be matched by $\context{Q}$. 
We argue that mobility, implemented by passing private channels as messages, is central to the $\pi$-calculus paradigm; hence this limitation of the notion of bisimilarity originally proposed for the applied $\pi$-calculus is significant.

The time now is right to move on from the above issue with mobility. 
The above limitation of the original conference version of bisimilarity~\cite{abadi01popl} has been addressed in a journal version~\cite{Abadi16}.
The trick is simple: allow channels to be messages. This way, a ``recipe'' to produce the channel name can used to indirectly refer to channels, such as $\fst{u}$ in the example context above, as permitted in ProVerif~\cite{Blanchet2008}.
Other notions of bisimilarity for the applied $\pi$-calculus~\cite{Delaune2007,Liu2010}, each following the old conference style for channels, can also be repaired by allowing messages as channels.
Note the $\psi$-calculus~\cite{Bengtson2009} was introduced as an alternative response to this problem; however the aforementioned journal paper~\cite{Abadi16} instead makes minimal changes necessary to repair the applied $\pi$-calculus (permitting channels to be messages).

\paragraph{Bisimilarity congruences for the applied $\pi$-calculus.}
This paper takes work on bisimilarity for the applied $\pi$-calculus a step further.
We explore notions of bisimilarity closed under any context, not just under any parallel context.
That is, we seek equivalences that are simultaneously a bisimilarity and a congruence. 
We identify three advantages of employing \textbf{bisimilarity congruences}:
\begin{itemize}
\item Improved \textit{algebraic reasoning}: since full compositionality is guaranteed, a process can replace an equivalent process anywhere inside a larger process term.

\item Improved \textit{robustness}: once two processes are proven to be equivalent, even if an attacker has the power to change the context of the process during runtime, an attack distinguishing the processes cannot be performed.

\item Improved \textit{state-space exploration}: bisimilarity congruences can instantiate inputs lazily, hence less resources are required to prove larger processes are equivalent.
\end{itemize}

There is a precedent for this investigation. A bisimilarity congruence called \textit{open bisimilarity} has been studied for the $\pi$-calculus~\cite{sangiorgi96acta}, and for a more restricted predecessor to the applied $\pi$-calculus called the spi-calculus~\cite{briais06entcs,tiu07aplas}.
The lazy approach to instantiating inputs made open bisimilarity the favoured bisimilarity for the mobility workbench~\cite{Victor94} --- the first toolkit to implement the $\pi$-calculus. 
Open bisimilarity has also been used in decision procedures for the spi-calculus~\cite{Tiu2010CSF}, by exploiting most general unifiers as finite representation of infinitely many inputs.



This work explores a notion of bisimilarity for the applied $\pi$-calculus, called \textit{open barbed bisimilarity}~\cite{Sangiorgi2001}.
Open barbed bisimilarity is a canonical choice, being the coarsest bisimilarity congruence for the applied $\pi$-calculus, i.e., a bisimilarity congruence with respect to which all bisimilarity congruences are sound.
The definition of open barbed bisimilarity is language independent, hence can be used to address unresolved design decisions.
We address two issues in particular. 
\begin{itemize}
\item Firstly, how can we define a bisimilarity congruence that can be used to reason about arbitrary message theories, not limited to xor~\cite{Ayala-Rincon2017} and blind signatures~\cite{Bursuc2014}?
Existing work on open bisimilarity for the spi-calculus is hard wired to handle Dolev-Yao~\cite{dolev83tit} symmetric encryption only.

\item Secondly, how do we define a bisimilarity congruence sufficiently coarse to verify privacy properties? 
Many privacy protocols involve \texttt{if-then-else} branching to provide dummy information to avoid privacy attacks via control flow analysis.
\end{itemize}
The latter problem is surprisingly subtle. Until recently, there was no definition of a bisimilarity that is a congruence and can handle \texttt{if-then-else} branching 
in the $\pi$-calculus, even without cryptographic primitives.
Recent work~\cite{Horne2018}, explains an approach to \texttt{if-then-else} in the $\pi$-calculus;
but with a warning: additional care must be taken to ensure the bisimilarity congruence can verify privacy properties.
Without care, an excessively lazy bisimilarity congruence, will claim to discover attacks that do not exist.

\paragraph{Subtle privacy properties.}
We illustrate a recurrent problem for verifying privacy protocols. The following example is a drastically cut down version of a classic private server example~\cite{Abadi2004,Cheval2017}, sufficient to explain the essence of the problem.
\[
\begin{array}{rl}
\mbox{Server A:}&\qquad \nu k. \nu r. \cout{a}{\pk{k}}.a(x).\cout{a}{r}
\\
\mbox{Server B:}&\qquad
\nu k. \nu r. \cout{a}{\pk{k}}.a(x).\texttt{if}\,x = \pk{k}\,\texttt{then}\,\cout{a}{\aenc{\pair{m}{r}}{\pk{k}}}\,\texttt{else}\,\cout{a}{r}
\end{array}
\]
Both processes above first transmit a public key, then receive a message.
Server A then transmits a random fresh name (a nonce) regardless of message received.
In contrast, Server B makes a decision based on the input.
If the input is the public key previously transmitted, then Server B responds with a message-nonce pair encrypted with the public key.
Otherwise, Server B sends a dummy random message, behaving as Server A.

Server A and Server B are indistinguishable to an external observer --- the attacker.
An attacker cannot learn that Server B responds in a special way to input $\pk{k}$ (the public key corresponding to private key $k$).
The idea is an attacker without private key $k$ cannot learn that Server B serves some data $m$ to the owner of $k$.
Thus the privacy of the intended recipient of the data is preserved.

We can verify this privacy property by showing Server A and Server B are bisimilar.
The warning is: we must take care about which bisimilarity congruence we employ. 
If, instead of open barbed bisimilarity, we employ the more famous \textit{open bisimilarity}~\cite{sangiorgi96acta},
the processes are \textbf{not} equivalent.
The law of excluded middle is invalidated for open bisimilarity~\cite{Ahn2017}; hence Server B can reach a state where it is not yet decided whether $x = \pk{k}$ or $x \not= \pk{k}$ at which point the \textit{if-then-else} branching cannot yet be resolved; but Server A cannot reach an equivalent state.
This distinguishing strategy, does not correspond to a real attack on the privacy of Server B; hence open bisimulation is not sufficiently coarse to verify this privacy property.

Fortunately, \textit{open barbed bisimilarity} address the above limitation of open bisimilarity. Open barbed bisimilarity is also intuitionistic, but private information, such as $\pk{k}$, is treated classically.
Thereby, after receiving the input either $x = \pk{k}$ or $x \not= \pk{k}$ holds; from which we can establish Server A and Server B are \textit{open barbed bisimilar}. 

\paragraph{Describing attacks.}
When there is a genuine attack it can be described using a modal logic formula.
The modal logic we propose is ``intuitionistic $\FM$'', which is proven in this work to logically characterise open barbed bisimilarity.
Thus whenever two processes are not open barbed bisimilar, we can construct a formula in intuitionistic $\FM$ 
that holds for one process only.

As an example of a distinguishing formula, consider a slight modification of Server B, without nonce $r$ in the encrypted message (from a cryptographic perspective this means a deterministic asymmetric encryption scheme is employed to encrypt $m$).
\[
\mbox{Server C:}\qquad
\nu k. \nu r. \cout{a}{\pk{k}}.a(x).\texttt{if}\,x = \pk{k}\,\texttt{then}\,\cout{a}{\aenc{m}{\pk{k}}}\,\texttt{else}\,\cout{a}{r}
\]
Now, Server C is not open barbed bisimilar to Server A.
The attack on the privacy of the protocol can be described by the following modal logic formula.
\[
\mbox{Server C} \vDash \diam{\co{a}(v)}\diam{a\,v}\diam{\co{a}(w)}\left( \aenc{m}{v} = w \right)
\]
The formula above is satisfied by Server C, but not by Server A (nor, by equivalence, Server B).
The attack described by the formula above is as follows: the attacker takes an output, named $v$, and feeds it back in as an input, then receives another output $w$.
At this point the attacker can reconstruct message $w$ using messages $v$ and $m$ (where $m$ is an open term representing a known or guessable plaintext).
Thus the attacker can determine that the server responds differently when the input received is $v$, i.e., message $\pk{k}$; hence the privacy of Server C is compromised.


All examples above, elaborated on in the body of the paper, are selected to be a minimal explanation to subtleties of \texttt{if-then-else} branching addressed by open barbed bisimilarity.

\paragraph{Summary.}
The body of the paper develops the theory of open barbed bisimilarity, as a robust foundation for verifying cryptographic protocols.
Section~\ref{section:api-open} introduces (strong) open barbed bisimilarity.
Section~\ref{section:quasi} introduces a variant of labelled bisimilarity called \textit{quasi-open bisimilarity} and proves that it coincides with open barbed bisimilarity.
Section~\ref{section:modal} defines an intuitionistic modal logic characterising quasi-open bisimilarity; hence also open barbed bisimilarity.
Section~\ref{section:privacy} provides more substantial examples of security and privacy properties.
Section~\ref{section:related} compares open barbed bisimilarity to other bisimilarities, including established notions of labelled bisimilarity.

\section{The Coarsest Bisimilarity Congruence for the Applied $\pi$-calculus}
\label{section:api-open}

This section concerns the coarsest (strong) bisimilarity congruence, open barbed bisimilarity.
Open barbed bisimilarity has not previously been explored for any cryptographic calculus.
However, it is a natural choice of bisimilarity, being, by definition, the greatest bisimilarity congruence.
Since open barbed bisimilarity has an objective language-independent definition, there are no design decisions --- there is only one reasonable definition as explored in this section.

\subsection{An example message term language and equational theory.}
\label{section:static}

In the applied $\pi$-calculus messages can be defined with respect to any message language subject to any equational theory ($=_{E}$).
The example we provide in Fig.~\ref{figure:messages} is for the purpose of meaningful examples.
Further theories can also be devised not limited to: 
sub-term convergent theories~\cite{Abadi2006};
blind signatures and homomorphic encryption~\cite{Bursuc2014}; and locally stable theories with inverses~\cite{Ayala-Rincon2017}.

\begin{figure}[h]
\[
\begin{gathered}
\begin{array}{rlr}
M, N, K \Coloneqq& x & \mbox{variable} \\
          \mid& \pk{M} & \mbox{public key} \\
          \mid& \hash{M} & \mbox{hash} \\
          \mid& \left\langle M, N\right\rangle & \mbox{tuple} \\
          \mid& \aenc{M}{N} & \mbox{encryption} \\
          \mid& \adec{M}{N} & \mbox{decryption} \\
          \mid& \fst{M} & \mbox{left} \\
          \mid& \snd{M} & \mbox{right} \\
\end{array} 
\qquad
\begin{array}{c}
\adec{\aenc{M}{\pk{K}}}{K} =_{E} M
\\[12pt]
\aenc{\adec{M}{K}}{\pk{K}} =_{E} M
\\[12pt]
\fst{\pair{M}{N}} =_{E} M
\\[12pt]
\snd{\pair{M}{N}} =_{E} N
\end{array}
\end{gathered}
\]
\caption{The applied $\pi$-calculus can be instantiated with \textbf{any} message language and equational theory for messages.
This example message theory is provided only for the purpose of providing meaningful examples.
}
\label{figure:messages}
\end{figure}

The example theory provided in Fig.~\ref{figure:messages} covers asymmetric encryption.
A message encrypted with public key $\pk{k}$ can only be decrypted using private key $k$.
The theory includes a collision-resistant hash function, with no equations.
This theory assumes we have the power to detect whether a message is a pair, but cannot distinguish a failed decryption from a random number.

\subsection{Active substitutions and open early transitions.}

We define the syntax of the applied $\pi$-calculus.
The syntax is similar to the $\pi$-calculus, except messages and channels can be any term rather than just variables.
There is no separate syntactic class of terms for names --- names are variables bound by new name binders.
In addition to processes, \textit{extended processes} are defined, which allow \textit{active substitutions} to float alongside processes and in the scope of new name binders, defined as follows.
\[
\begin{gathered}
\begin{array}{rlr}
P, Q \Coloneqq& 0 & \mbox{deadlock} \\
          \mid& \cout{M}{N}.P & \mbox{send} \\
          \mid& \cin{M}{y}.P & \mbox{receive} \\
          \mid& \mathopen{\left[M = N\right]}P & \mbox{match} \\
          \mid& \mathopen{\left[M \not= N\right]}P & \mbox{mismatch} \\
          \mid& \mathopen\nu x. P & \mbox{new} \\
          \mid& P \cpar Q & \mbox{parallel} \\
          \mid& P + Q & \mbox{choice} \\
          \mid& \bang P & \mbox{replication} 
\end{array}
\qquad
\begin{array}{l}
\mbox{Extended processes in normal form:}
\\[5pt]
\begin{array}{rlr}
A, B  \Coloneqq& \sigma \cpar P & \mbox{process with active substitution} \\ 
          \mid& \mathopen{\nu x.} A & \mbox{new} \\
\end{array}
\\[20pt]
\mbox{actions on labels:}
\\[5pt]
\begin{array}{rlr}
\pi \Coloneqq & \tau      & \mbox{progress} \\
         \mid & \co{M}(z) & \mbox{bound output} \\
         \mid & M\,N      & \mbox{free input} \\
\end{array}
\end{array}
\end{gathered}
\]
Extended processes in normal form $\mathopen{\nu \vec{x}.}\left(\sigma \cpar P\right)$ are subject to the restriction
that the variables in $\dom{\sigma}$ are fresh for $\vec{x}$, $\fv{P}$ and $\fv{y\sigma}$, for all variables $y$ (i.e., $\sigma$ is idempotent, and substitutions are fully applied to $P$).
We follow the convention that 
operational rules are defined directly on extended processes in normal forms.
This avoids numerous complications caused by the structural congruence in the original definition of bisimulation for the applied $\pi$-calculus.
We require the following definitions for composing extended processes in parallel and with substitutions,
defined whenever $z\not\in \fv{B} \cup \fv{\rho}$ and $\dom{\sigma} \cap \dom{\theta} = \emptyset$.
\begin{gather*}
\mathopen{\nu z.} A \cpar B
\triangleq
\mathopen{\nu{z}.} \left( A \cpar B \right)
\quad
B \cpar \mathopen{\nu z.} A
\triangleq
\mathopen{\nu{z}.} \left( B \cpar A \right)
\qquad
(\sigma \cpar P) \cpar (\theta \cpar Q) \triangleq \sigma\cdot\theta \cpar (P \cpar Q)
\\
\rho \cpar \nu z.A \triangleq \mathopen{\nu z.}\left(\rho \cpar A\right)
\qquad
\sigma \cpar \theta \cpar Q
\triangleq
\sigma\cdot\theta \cpar Q
\end{gather*}

\textit{Intuitionistic mismatch.}
Mismatch requires special attention. Mismatch models the \texttt{else} branch of an \texttt{if-then-else} statement with an equality guard.
We define $\texttt{if}\,M=N\,\texttt{then}\,P\,\texttt{else}\,Q$ as an abbreviation for $\match{M=N}P + \match{M\not=N}Q$.

As uncovered in related work~\cite{Horne2018},
the trick for handling mismatch such that we obtain a congruence is to treat mismatch intuitionistically.
Intuitionistic negation enjoys the property that it is preserved under substitutions; a property that fails for classical negation in general.
E.g., there are substitutions under which $\match{x \not= \hash{y}}a(z)$ can perform an input transition and others where it cannot, hence neither $x = \hash{y}$ nor $x \not= \hash{y}$ holds in the intuitionistic setting, until more information is provided about the environment.
In order to define intuitionistic negation, we require the notion of a fresh substitution; which is also critical for the logical characterisation introduced later in Section~\ref{section:modal}.
\begin{definition}[fresh]\label{def:respects}
Given a set of variables $\env $ and substitution $\sigma$,
we say $\sigma$ is fresh for $\env$
whenever $\dom{\sigma} \cap \env = \emptyset$, and, for all $y \not\in \env $, we have $\fv{y\sigma} \cap \env  = \emptyset$.
We say entailment $\env  \vDash M \not= N$ holds whenever
there is no $\sigma$ fresh for $\env $ such that $M\sigma =_{E} N\sigma$.
\end{definition}
Consider the following examples that hold or fail to hold for different reasons.
Entailment $\emptyset \vDash x \not= h(x)$ holds, since there exists no unifier, witnessed by a simple occurs check.
In contrast, $\emptyset \vDash x \not= h(y)$ does not hold, since there exists substitution $\sub{x}{h(y)}$ unifying messages $x$ and $h(y)$, so it is still possible the messages could be equal; thus, there is insufficient information to decided whether the messages are equal or not.
By extending the environment such that $y$ is a private name, entailment $y \vDash x \not= h(y)$ holds, since most general unifier $\sub{x}{h(y)}$ is not fresh for $\left\{ y \right\}$ --- an observer who can influence $x$, cannot make $x$ equal to $h(y)$ without access to $y$.

To define open barbed bisimulation, we require an \textit{open early labelled transition system} for the applied $\pi$-calculus in Fig.~\ref{figure:active}.
There are three types of label: $\tau$ representing some internal progress due to communication; bound output $\co{M}(x)$ representing that something bound to $x$ is sent on channel $M$;
and free input $M\,N$ representing message $N$ is received on channel $M$.

\begin{figure}
\[
\begin{gathered}
\begin{array}{c}
\begin{prooftree}
\justifies
\env  \colon \mathopen{\cin{M}{x}.}P \lts{M\,N} {P\sub{x}{N}}
\using
\mbox{\textsc{Inp}}
\end{prooftree}
\quad
\begin{prooftree}
x \not \in \fv{M}\cup\fv{N}\cup\fv{P} \cup \env 
\justifies
\env  \colon \cout{M}{N}.P \lts{\co{M}(x)} \sub{x}{N} \cpar P
\using
\mbox{\textsc{Out}}
\end{prooftree}
\quad
\begin{prooftree}
 \env  \colon P \lts{\pi} A
\justifies
 \env  \colon P + Q \lts{\pi} A
\using
\mbox{\textsc{Sum-l}}
\end{prooftree}
\\[15pt]
\begin{prooftree}
 \env  \colon P \lts{\pi} A
\justifies
 \env  \colon {\mathopen{\left[M=M\right]}{P}\lts{\pi}{A}}
\using
\mbox{\textsc{Mat}}
\end{prooftree}
\qquad\qquad
\begin{prooftree}
 \env  \colon P \lts{\pi} A
\qquad
 \env  \vDash M \not= N
\justifies
 \env  \colon {\mathopen{\left[M \not= N\right]}{P}\lts{\pi}{A}}
\using
\mbox{\textsc{Mismatch}}
\end{prooftree}
\\[15pt]
\begin{prooftree}
\env , x \colon A \lts{\pi} B
\quad
x \not\in \mathrm{n}(\pi) \cup \env 
\justifies
 \env  \colon {{\nu x.A}\lts{\pi}{\nu x.B}}
\using
\mbox{\textsc{Res}}
\end{prooftree}
\qquad\qquad
\begin{prooftree}
 \env  \colon {P \lts{\pi} A}
\quad
\mbox{$\bn{\pi} \cap \fv{Q} = \emptyset$}
\justifies
 \env  \colon {{P \cpar Q} \lts{\pi} {A \cpar Q}}
\using
\mbox{\textsc{Par-l}}
\end{prooftree}
\\[15pt]
\begin{prooftree}
 \env  \colon {P \lts{\pi\sigma} A}
\quad
\mbox{$\sigma$ fresh for $\bn{\pi}$}
\justifies
 \env  \colon {{\sigma \cpar P} \lts{\pi} {\sigma \cpar A}}
\using
\mbox{\textsc{Alias}}
\end{prooftree}
\qquad\qquad\qquad
\begin{prooftree}
\env  \colon P \lts{\pi} A
\justifies
\env  \colon \bang P \lts{\pi} A \cpar \bang P
\using
\mbox{\textsc{Rep-act}}
\end{prooftree}
\\[15pt]
\begin{prooftree}
 \env  \colon P \lts{\co{M}(x)} \nu \mathopen{\vec{z}.}\left(\sub{x}{N} \cpar P'\right)
\qquad
 \env  \colon Q \lts{M\,N} Q' 
\qquad
 \left(\left\{x \right\} \cup \vec{z}\right) \cap \fv{Q} = \emptyset
\justifies
 \env  \colon {P \cpar Q}\lts{\tau}{\mathopen{\nu \vec{z}.}\left(P' \cpar Q'\right)} 
\using
\mbox{\textsc{Close-l}}
\end{prooftree}
\\[15pt]
\begin{prooftree}
\env  \colon P \lts{\co{M}(x)} \mathopen{\nu \vec{z}.}\left(\sub{x}{N} \cpar Q\right)
\qquad
\env  \colon P \lts{M\,N} R
\qquad
 \vec{z} \cap \fv{P} = \emptyset
\justifies
\env  \colon \bang P \lts{\tau} \mathopen{\nu \vec{z}.}\left( Q \cpar R \cpar \bang P\right) 
\using
\mbox{\textsc{Rep-close}}
\end{prooftree}
\end{array}
\end{gathered}
\]
\caption{An \textit{open early} labelled transition system, plus symmetric rules for parallel composition and choice.
The equational theory over message terms can be applied to equate the occurrences of $M$ in the rules, \textsc{Inp}, \textsc{Out}, \textsc{Mat}, \textsc{Close-l}, and \textsc{Rep-close}.
The set of free variables and $\alpha$-conversion are as standard, where $\nu x.P$ and $M(x).P$ bind $x$ in $P$. 
Define the bound names such that $\bn{\pi} = \left\{ x \right\}$ only if $\pi = \co{M}(x)$ and $\bn{\pi} = \emptyset$ otherwise.
Define the names such that $\n{M\,N} = \fv{M} \cup \fv{N}$, $\n{M(x)} = \fv{M} \cup \left\{x\right\}$ and $\n{\tau} = \emptyset$.
}\label{figure:active}
\end{figure}

\textbf{The \textsc{Mismatch} and \textsc{Res} rules.}
The \textsc{Mismatch} rule is defined in terms of the entailment relation in Def.~\ref{def:respects}. 
The \textsc{Res} rule can also influence mismatches by introducing fresh private names.
For example, the following derivation shows an input transition is enabled.
\[
\begin{prooftree}
\begin{prooftree}
\begin{prooftree}
\justifies
y \colon z(w) \lts{z\,w} 0
\using
\textsc{Inp}
\end{prooftree}
\qquad
y \vDash x \not= h(y)
\justifies
y \colon \match{x \not= h(y)}z(w) \lts{z\,w} 0
\using
\textsc{Mismatch}
\end{prooftree}
\justifies
\emptyset \colon
\nu y. \match{x \not= h(y)}z(w) \lts{z\,w} \nu y. 0
\using
\textsc{Res}
\end{prooftree}
\] 
Notice, the bound variable $y$ is added to the set of private names, enabling entailment $y \vDash {x \not= h(y)}$.

\textbf{The \textsc{Alias} rule.}
A special alias rule is used in this normal-form presentation of the applied $\pi$-calculus; serving the purpose of applying active substitutions, while avoiding problems caused by the structural congruence.
For a non-trivial example of the \textsc{Alias} rule, \textsc{Res} rule and equational theory working together observe the following transition is  derivable.
\[
\begin{prooftree}
\begin{prooftree}
\begin{prooftree}
\justifies
m \colon
m(x)
\lts{\fst{\pair{m}{n}}\,x}
0
\using
\textsc{Inp}
\end{prooftree}
\justifies
m \colon
\sub{w}{\pair{m}{n}} \cpar m(x)
\lts{\fst{w}\,x}
\sub{w}{\pair{m}{n}} \cpar 0
\using
\textsc{Alias}
\end{prooftree}
\justifies
\emptyset \colon
\mathopen{\nu m.}\left( \sub{w}{\pair{m}{n}} \cpar m(x) \right)
\lts{\fst{w}\,x}
\mathopen{\nu m.}\left( \sub{w}{\pair{m}{n}} \cpar 0 \right)
\using
\textsc{Res}
\end{prooftree}
\]
The conditions on the \textsf{Res} rule ensure bound name $m$ cannot appear in the terms on the label.
Fortunately, the \textsc{Alias} rule allows a $m$ to be expressed in terms of extruded variable $w$.
Since $m =_E \fst{\pair{m}{n}}$ and the equational theory can be applied in rule \textsc{Inp}, the above input action on channel $\fst{w}$ is enabled,
indirectly representing that channel $m$ is used for the input action.

Note a device with the same effect as the \textsc{Alias} rule is used in the proof of the recently corrected definition of \textit{labelled bisimilarity}~\cite{Abadi16}. Note in particular the normal form presentations of labelled transitions in the definition between B.9 and B.10 in the extended Arxiv version of the same paper~\cite{Arxiv}.
A normal form presentation is also used in ProVerif; hence there should be no controversy employing normal forms and the \textsc{Alias} rule.

\textbf{The \textsc{Out} rule.}
A rule differing significantly from standard presentations of the core $\pi$-calculus is the \textsc{Out} rule.
Instead of recording the message sent on the label, the message is recorded in an active substitution. The domain of the active substitution is chosen to be a fresh variable appearing as the bound variable in the output action on the label.

In the following example a message is sent using the \textsc{Out} rule, then the \textsc{Res} rule is applied such that the private name $n$ in the active substitution appears bound before and after the transition.
\[
\begin{prooftree}
\begin{prooftree}
\justifies
n, k \colon
\cout{a}{\aenc{n}{\pk{k}}}.n(x) 
\lts{\co{a}(w)} {\sub{w}{\aenc{n}{\pk{k}}}} \cpar n(x)
\end{prooftree}
\justifies
k \colon
\mathopen{\nu n.}\cout{a}{\aenc{n}{\pk{k}}}.n(x) 
\lts{\co{a}(w)} \mathopen{\nu n.}\left(\sub{w}{\aenc{n}{\pk{k}}} \cpar n(x)\right)
\end{prooftree}
\]
Observe, by rule \textsc{Inp}, the following input action is enabled.
\[
k \colon
{a(w).\cout{\adec{w}{k}}{a}} \lts{a\,\aenc{n}{\pk{k}}} \cout{\adec{\aenc{n}{\pk{k}}}{k}}{a}
\]
Hence by \textsc{Close-l} the following interaction is enabled,
using $\adec{\aenc{n}{\pk{k}}}{k} =_{E} n$.
\[
k \colon
{\mathopen{\nu n.}\cout{a}{\aenc{n}{\pk{k}}}.n(x) 
\cpar
a(w).\cout{\adec{w}{k}}{a}}
\lts{\tau}
\mathopen{\nu n.}\left(
n(x)
\cpar
\cout{n}{a}
\right)
\]
Note this labelled approach to interaction follows closely how interaction traditionally works in the $\pi$-calculus.
An advantage of our labelled transition approach is \textit{strong} and \textit{weak} variants of bisimilarities can be studied.
In contrast, the original system proposed for the applied $\pi$-calculus~\cite{abadi01popl} used a hybrid labelled/reduction system that can only be used to formalise \textit{weak} bisimilarities.

\begin{definition}
As a convention, write $A \lts{\pi} B$ whenever $\emptyset \colon A \lts{\pi} B$.
\end{definition}

\subsection{An objective bisimilarity congruence: open barbed bisimilarity.}

A barb represents the ability to observe an input or output action on a channel.
Barbs are typically used to define \textit{barbed equivalence}, or \textit{observational equivalence}~\cite{Milner1992}.
However, barbed equivalence is a congruence but not a bisimilarity; while observational equivalence is a bisimilarity but not a congruence. 
For this reason, we prefer \textit{open barbed bisimilarity}~\cite{Sangiorgi2001},
which is, by definition, both a bisimilarity and a congruence.
\begin{definition}[open barbed bisimilarity]\label{def:open-barbed}
A process $P$ has barb $M$, written $\barb{P}{M}$,
whenever,
for some $A$, 
$P \lts{\co{M} (z)} A$,
or $P \lts{M\,N} A$.
An open barbed bisimulation $\rel$ is a symmetric relation over processes
such that whenever $P \rel Q$ holds 
the following hold:
\begin{itemize}
\item For \textit{all} contexts $\context{\ \cdot\ }$, $\context{P} \rel \context{Q}$.
\item If $\barb{P}{M}$ then $\barb{Q}{M}$.
\item If $P \lts{\tau} P'$, there exists $Q'$ such that $Q \lts{\tau} Q'$ and $P' \rel Q'$ holds.
\end{itemize}
Open barbed bisimilarity $\bsim$ is the greatest open barbed bisimulation.
\end{definition}
The power of open barbed bisimilarity comes from
closing by all contexts at every step, not only at the beginning of execution.
Closing by all contexts at every step ensures the robustness of open barbed bisimilarity even if the environment changes at runtime;
i.e.,\ we stay within a congruence relation at every step of the bisimulation game.

Recall a congruence is an equivalence relation closed under all contexts.
Symmetry and context closure are immediate from definition of open barbed bisimilarity. Reflexivity is trivial since the identity relation over extended processes is an open barbed bisimulation.
Transitivity is slightly more involved, proven by showing that the transitive closure of two open barbed bisimulations is an open barbed bisimulation.

Open barbed bisimilarity is concise --- the definition requires only the open labelled transition system in Fig.~\ref{figure:active} and the three clauses in Definition~\ref{def:open-barbed}.
Furthermore, objectively, open barbed bisimilarity is the coarsest bisimilarity congruence,
in the sense that it is by definition a congruence, and defined independently of the content of the messages sent and received.
Notice, due to the independence of the information on the labels, open barbed bisimilarity applies to any language;
indeed open barbed bisimilarity is a generalisation of dynamic observational equivalence~\cite{Sassone1992}, that, historically, was used to objectively identify the greatest bisimulation congruence for \textsf{CCS}.

For the above reasons, open barbed bisimilarity is an ideal reference definition.
However, as with all barbed congruences it is unwieldy due to the closure under all contexts.
This leads us to the notion of quasi-open bisimilarity in the next section which is easier to use.

\section{Quasi-open bisimilarity for the applied $\pi$-calculus}
\label{section:quasi}

As highlighted in the previous section, open barbed bisimilarity is concise to define but difficult to check, due to the quantification over all contexts.
An open variant of labelled bisimilarity, called \textit{quasi-open bisimilarity}, avoids quantifying over all contexts;
and furthermore, coincides with open barbed bisimilarity.
In this section, we lift quasi-open bisimilarity to the setting of the applied $\pi$-calculus,
generalising established results for the $\pi$-calculus~\cite{Sangiorgi2001}
and the $\pi$-calculus with mismatch~\cite{Horne2018}.

\subsection{Recalling the standard definition of static equivalence.}
To extend quasi-open bisimulation to the applied $\pi$-calculus the notion of static equivalence is required. 
Static equivalence is defined over the static information in an extended process --- the active substitutions and name restrictions.
\begin{definition}[static equivalence]\label{def:static}
Two normal form extended processes 
$\mathopen{\nu\vec{x}.}\left( \sigma \cpar P \right)$ 
and 
$\mathopen{\nu\vec{y}.}\left( \theta \cpar Q \right)$ 
are statically equivalent
whenever 
for all messages $M$ and $N$ such that $\left(\fv{M} \cup \fv{N}\right) \cap \left(\vec{x} \cup \vec{y}\right) = \emptyset$,
$M\sigma =_{E} N\sigma$ if and only if $M\theta =_{E} N\theta$.
\end{definition}
In the above definition, messages $M$ and $N$ represent to different ``recipes'' for producing messages.
Two extended processes are distinguished by static equivalence only when the two recipes produce equivalent messages under one substitution, but distinct messages under the other substitution.

\paragraph{Static equivalence examples.}
The concept of static equivalence is no different from original work on the applied $\pi$-calculus~\cite{Abadi16}. However, for a self-contained presentation we provide examples.
The following extended processes are not statically equivalent.
\[
\mathopen{\nu m,n.}\left(\sub{v,w}{m,n} \cpar 0\right)
\qquad
\mbox{v.s.}
\qquad
\mathopen{\nu m.}\left(\sub{v,w}{m,\hash{m}} \cpar 0\right)
\]
They are distinguished by messages $\hash{v}$ and $w$.
To see why, $\hash{v}\sub{v,w}{m,\hash{m}}$ and $w\sub{v,w}{m,\hash{m}}$ are both equal to $\hash{m}$; but $\hash{v}\sub{v,w}{m,n}$ is distinct from $w\sub{v,w}{m,n}$.

For a less obvious example, consider the following extended processes.
\[
\mathopen{\nu m, k, n.}\left( \sub{x_1, x_2}{\aenc{m}{\pk{k}}, n} \cpar 0 \right)
\qquad
\mbox{v.s.}
\qquad
\mathopen{\nu m, k.}\left( \sub{x_1, x_2}{\aenc{m}{\pk{k}}, k} \cpar 0 \right)
\]
Perhaps surprisingly, the above extended processes are statically equivalent. 
This relies on the fact that the example message theory, in Fig.~\ref{figure:messages}, does not allow successful decryption to be detected.
This assumption about asymmetric encryption avoids common problems, including Bleichenbacher's vulnerability on SSL~\cite{Bleichenbacher1998}.
Thus, for example, recipe $\adec{x_1}{x_2}$ produces what looks like a random number for both processes.

If a protocol requires successful decryption to be detected, entropy should be introduced when a nonce is encrypted.
For example, consider the following extended processes, where the nonce $m$ is tagged with $t$ before being encrypted.
\[
\mathopen{\nu m, k, n.}\left( \sub{x_3,x_4}{\aenc{\pair{t}{m}}{\pk{k}}, n} \cpar 0 \right)
\qquad
\mbox{v.s.}
\qquad
\mathopen{\nu m, k.}\left( \sub{x_3,x_4}{\aenc{\pair{t}{m}}{\pk{k}}, k} \cpar 0 \right)
\]
In contrast to the previous example, the above are not statically equivalent.
The above processes can be distinguished by recipes $\fst{\adec{x_3}{x_4}}$ and $t$, which  only produce equal messages according to the extended process above right, in contrast to the extended process above on the left.

\subsection{Introducing the new definition of quasi-open bisimilarity.}
For an elegant definition of quasi-open bisimilarity, we employ the following reachability relation.
\begin{definition}[reachability]\label{definition:reachable}
Given extended processes $A$ and $B$, we say $A$ can reach $B$ by substitution $\sigma$ and
environment extension $\vec{z}.\rho$, written $\reachable{A}{\sigma,\nu\vec{z}.\rho}{B}$,
whenever:
$A$ is of the form $\mathopen{\nu \vec{x}.}\left(\theta \cpar P\right)$,
$\sigma$ and $\rho$ are idempotent and fresh for $\dom{\theta}$,
$\vec{z} \cap \left(\dom{\rho} \cup \dom{\theta}\right) = \emptyset$
and
 $B = \mathopen{\nu \vec{z}.}\left(\rho \cpar A\sigma\rho \right)$.
\end{definition}
The first substitution in the definition above allows free variables to be instantiated. For example\footnote{The environment extension can be omitted when it is the identity extension: $\nu \emptyset.id$.},
$\match{z = \hash{x}}\tau \leq_{\sub{z}{\hash{x}}} \match{\hash{x} = \hash{x}}\tau$ indicates variable $z$ in process $\match{z = \hash{x}}\tau$, can take on value $\hash{x}$, thereby reaching a process where the match guard is enabled.
In contrast, since, as standard we assume substitution is capture avoiding, when we apply the same substitution to $\mathopen{\nu x.}\left( {\sub{u}{\hash{x}}} \cpar \match{z = \hash{x}}\tau \right)$, the bound name $x$ is renamed to $x'$ to avoid a clash with variable $x$ in the range of the substitution:
$\mathopen{\nu x.}\left( {\sub{u}{\hash{x}}} \cpar \match{z = \hash{x}}\tau \right)
\leq_{\sub{z}{\hash{x}}}
\mathopen{\nu x'.}\left( {\sub{u}{\hash{x'}}} \cpar \match{\hash{x} = \hash{x'}}\tau \right)$.
In this example, no substitution can enable the match guard.

Environment extensions are used to distinguish pairs of messages in a mismatch, in scenarios where neither message is ground.
For example, the following process $\match{x \not= z}a(y).\match{x = y}\tau$ cannot yet act --- there is no evidence $x$ and $z$ are distinct.
However, by extending the environment with a fresh name, we have $\mathopen{\nu x.}\left({\sub{u}{x}} \cpar \match{x \not= z}a(y).\match{x = y}\tau\right) \lts{a\,u} \mathopen{\nu x.}\left({\sub{u}{x}} \cpar \match{x = x}\tau\right)$.
Notice alternative environment extensions also enable an input transition, such as the following: $\mathopen{\nu z.}\left({\sub{v}{z}} \cpar \match{x \not= z}a(y).\match{x = y}\tau\right) \lts{a\,x} \mathopen{\nu z.}\left({\sub{v}{z}} \cpar \match{x = x}\tau\right)$. 
Hence there is no unique most general substitution and set of names permanently distinguishing $x$ from $z$.

In order to define quasi-open bisimilarity, we require the notion of an \textit{open relation} between extended processes.
An open relation is preserved under reachability, defined above.
\begin{definition}[open]\label{def:close}
A relation over extended processes $\mathcal{R}$ is \textit{open}, whenever
if $A \mathrel{\mathcal{R}} B$ and $A \leq_{\sigma, \nu \vec{z}.\rho} A'$ and $B \leq_{\sigma, \nu \vec{z}.\rho} B'$,
then $A' \mathrel{\mathcal{R}} B'$.
\end{definition}

Given the definition of an open relation, static equivalence, and the labelled transition system, we can provide the following concise definition of quasi-open bisimilarity for the applied $\pi$-calculus.
\begin{definition}[quasi-open bisimilarity]\label{definition:quasi-open}
An \textbf{open} symmetric relation between extended processes $\mathrel{\mathcal{R}}$ is a quasi-open bisimulation whenever,
if $A \mathrel{\mathcal{R}} B$ then the following hold:
\begin{itemize}
\item $A$ and $B$ are statically equivalent.
\item If $A \lts{\pi} A'$ there exists $B'$ such that $B \lts{\pi} B'$ and $A' \mathrel{\mathcal{R}} B'$.
\end{itemize}
Quasi-open bisimilarity $\sim$ is the greatest quasi-open bisimulation.
\end{definition}

The keyword in the definition above is ``open'' in the sense of Def.~\ref{def:close}.
Without ensuring properties are preserved under reachability, the above definition would simply be the strong version of \textit{labelled bisimilarity} for the applied $\pi$-calculus~\cite{Abadi16}.
We illustrate the impact of insisting on an open relation and allowing messages as channels in the following examples.

We remark that 
the definition of quasi-open bisimilarity above is arguably simpler than in the original setting of the $\pi$-calculus~\cite{Sangiorgi2001}. In contrast to the original definition, since private names are recorded in extended processes, all types of action are handled by one clause and there is no need to index a bisimulation with extruded private names.

\paragraph{Mobility example.}
Recall processes $\nu z. \cout{x}{z,y}.z(w)$
and $\nu z. \cout{x}{z,y}$ from the introduction.
Recall these processes are \textbf{mistakenly} bisimilar according to the original definition of bisimilarity for the applied $\pi$-calculus~\cite{abadi01popl}.
These processes should not be equivalent, indeed they are valid polyadic $\pi$-calculus processes (the $\pi$-calculus with tuples)~\cite{Milner1993}, and the applied $\pi$-calculus should be conservative with respect the polyadic $\pi$-calculus.

Fortunately, these processes are correctly distinguished by quasi-open bisimilarity. To see why, firstly, consider the following two transitions with matching actions.
\[
\nu z.\cout{x}{z,y}.z(w) \lts{\co{x}(v)} \mathopen{\nu z.}\left({\sub{v}{\pair{z}{y}}} \cpar z(w)\right)
\qquad\mbox{and}\qquad
{\nu z.\cout{x}{z,y}} \lts{\co{x}(v)} \mathopen{\nu z.}\left({\sub{v}{\pair{z}{y}}} \cpar 0\right)
\]
The trick now is to use the \textsc{Alias} rule to enable the following labelled transition for the process on the left:
$\mathopen{\nu z.}\left({\sub{v}{\pair{z}{y}}} \cpar z(w)\right) 
\lts{\fst{v}\,x} \mathopen{\nu z.}\left({\sub{v}{\pair{z}{y}}} \cpar 0 \right)$.
The other process $\mathopen{\nu z.}\left({\sub{v}{\pair{z}{y}}} \cpar 0\right)$ is deadlocked, so cannot match this transition. Notice the use of message $\fst{v}$ as a channel.

\paragraph{Example showing impact of an open relation on static equivalence.}
By insisting that a quasi-open bisimulation is an open relation (Def.~\ref{def:close}), static equivalence must also be preserved by all fresh substitutions.
This has an impact on examples such as the following.

Consider for example the processes $\nu x.\cout{a}{\aenc{x}{z}}$ and $\nu x.\cout{a}{\aenc{\pair{x}{y}}{z}}$ that are labelled bisimilar but not quasi-open bisimilar.
To see why, firstly, observe both process can perform a $\co{a}(v)$-transition to the respective extended processes $\mathopen{\nu x.}\left({\sub{v}{\aenc{x}{z}}} \cpar 0 \right)$ and $\mathopen{\nu x.}\left({\sub{v}{\aenc{\pair{x}{y}}{z}}} \cpar 0 \right)$. Note that these extended process are \textbf{statically} equivalent.
However, since a quasi-open bisimulation must be preserved under fresh substitutions and $\sub{z}{\pk{w}}$ is fresh for $\left\{v\right\}$,
we should also check $\mathopen{\nu x.}\left({\sub{v}{\aenc{x}{z}}} \cpar 0\right)\sub{z}{\pk{w}}$ and $\mathopen{\nu x.}\left({\sub{v}{\aenc{\pair{x}{y}}{z}}} \cpar 0\right)\sub{z}{\pk{w}}$.
After applying the substitution, the extended processes are no longer statically equivalent, witnessed by distinguishing 
recipes $\snd{\adec{v}{w}}$ and $y$.
Thus the processes are not quasi-open bisimilar.


\paragraph{Example of privacy property.} 
We now have the mechanisms to verify the minimal privacy example from the introduction.
We prove the following by constructing a quasi-open bisimulation. 
\[\footnotesize
{\nu k. \nu r. \cout{a}{\pk{k}}.a(x).\cout{a}{r}}
\sim
\nu k. \nu r. \cout{a}{\pk{k}}.a(x).\texttt{if}\,x = \pk{k}\,\texttt{then}\,\cout{a}{\aenc{\pair{m}{r}}{\pk{k}}}\,\texttt{else}\,\cout{a}{r}
\]
Define quasi-open bisimulation $\mathcal{S}$ to be the least open symmetric relation such that: for all $M$ and $N$ fresh for $\left\{k,r\right\}$ and $u$ fresh for $\left\{ a, k, r \right\} \cup \fv{M} \cup \fv{N}$ and $v$ fresh for $\left\{ a, k, r, u \right\} \cup \fv{M} \cup \fv{N}$.
{
\[\scriptsize 
\begin{array}{rcl}
{\nu k. \nu r. \cout{a}{\pk{k}}.a(x).\cout{a}{r}}
&\mathrel{\mathcal{S}}&
\nu k. \nu r. \cout{a}{\pk{k}}.a(x).\texttt{if}\,x = \pk{k}\,\texttt{then}\,\cout{a}{\aenc{\pair{M}{r}}{\pk{k}}}\,\texttt{else}\,\cout{a}{r}
\\
\mathopen{\nu k. \nu r.} \left( \sub{u}{\pk{k}} \cpar a(x).\cout{a}{r} \right)
&\mathrel{\mathcal{S}}&
\mathopen{\nu k. \nu r.} \left( \sub{u}{\pk{k}} \cpar a(x).\texttt{if}\,x = \pk{k}\,\texttt{then}\,\cout{a}{\aenc{\pair{M}{r}}{\pk{k}}}\,\texttt{else}\,\cout{a}{r} \right)
\\
\mathopen{\nu k. \nu r.} \left( \sub{u}{\pk{k}} \cpar \cout{a}{r} \right)
&\mathrel{\mathcal{S}}&
\mathopen{\nu k. \nu r.} \left( \sub{u}{\pk{k}} \cpar \texttt{if}\,N\sub{u}{\pk{k}} = \pk{k}\,\texttt{then}\,\cout{a}{\aenc{\pair{M}{r}}{\pk{k}}}\,\texttt{else}\,\cout{a}{r} \right)
\\
\mathopen{\nu k. \nu r.} \left( \sub{u, v}{\pk{k}, r} \cpar 0 \right)
&\mathrel{\mathcal{S}}&
\mathopen{\nu k. \nu r.} \left( \sub{u, v}{\pk{k}, \aenc{\pair{M}{r}}{\pk{k}}} \cpar 0 \right)
\\
\mathopen{\nu k. \nu r.} \left( \sub{u, v}{\pk{k}, r} \cpar 0 \right)
&\mathrel{\mathcal{S}}&
\mathopen{\nu k. \nu r.} \left( \sub{u, v}{\pk{k}, r} \cpar 0 \right)
\end{array}
\]}
Critically, message $N$ ranges over all permitted inputs. For $N =_E u$, we have the following pair in relation $\mathcal{S}$. Observe the branch sending an encrypted message is enabled.
\[\footnotesize
\mathopen{\nu k. \nu r.} \left( \sub{u}{\pk{k}} \cpar \cout{a}{r} \right)
\mathrel{\mathcal{S}}
\mathopen{\nu k. \nu r.} \left( \sub{u}{\pk{k}} \cpar \texttt{if}\,\pk{k} = \pk{k}\,\texttt{then}\,\cout{a}{\aenc{\pair{M}{r}}{\pk{k}}}\,\texttt{else}\,\cout{a}{r} \right)
\]
If $N$ is any term not equivalent to $u$ then we have $k, r \vDash N\sub{u}{\pk{k}} \not= \pk{k}$ since if $N$ were a message term fresh for $\left\{k,r\right\}$ such that $N\sub{u}{\pk{k}} = \pk{k}$, then $N$ must be equivalent to $u$. Thus in all other cases the else branch is enabled.

Also, remark $\mathopen{\nu k. \nu r.} \left( \sub{u, v}{\pk{k}, r} \cpar 0 \right)$
and
$\mathopen{\nu k. \nu r.} \left( \sub{u, v}{\pk{k}, \aenc{\pair{M}{r}}{\pk{k}}} \cpar 0 \right)$ (reachable when $N =_E u$) 
are statically equivalent.
An attacker neither has the key $k$ to decrypt $\aenc{\pair{M}{r}}{\pk{k}}$,
nor can an attacker reconstruct the message $\pair{M}{r}$, without knowing $r$.

\subsection{Guaranteed fully compositional reasoning, including under input prefixes.}\label{section:compositionality}
We illustrate how quasi-open bisimilarity improves compositionality guarantees.
Consider the following two processes, which are sub-terms of Server A and Server C from the introduction.
\[
A' \triangleq {\cout{a}{r}}
\qquad
\mbox{v.s.}
\qquad
C' \triangleq 
\texttt{if}\,x = \pk{k}\,\texttt{then}\,\cout{a}{\aenc{m}{\pk{k}}}\,\texttt{else}\,\cout{a}{r}
\]
Now ask the question: should the above processes be equivalent or distinguished?
With respect to the original labelled bisimilarity proposed for the applied $\pi$-calculus~\cite{Abadi16}, the above processes are labelled bisimilar.
Reasoning classically, $x$ and $\pk{k}$ cannot be equal, since $x$ and $k$ are treated as ground terms rather than variables.
By this reasoning, the \texttt{else} branch in process $C'$ is enabled. 
The \texttt{else} branch in process $C'$ behaves as process $A'$, on the left above; hence, classically, the above processes are equivalent according labelled bisimilarity.

Na{\"i}vely, it may be tempting at this point to attempt to reason compositionally, with respect to 
common context $\context{\ \cdot\ } \triangleq \mathopen{\nu k. \nu r. \cout{a}{\pk{k}}.a(x).}\left\{\ \cdot\ \right\}$. 
However, for this example, $\context{ A' }$ (Server A) is not bisimilar to $\context{ C' }$ (Server C), according to labelled bisimilarity (nor quasi-open bisimilarity).
Hence such reasoning with respect to labelled bisimilarity and context $\context{\ \cdot\ }$ is \textbf{unsound}.

Such failures of compositionality for labelled bisimilarity with respect to contexts are addressed by quasi-open bisimilarity.
As we verify in the next section, whenever two processes are proven to be quasi-open bisimilar, they are quasi-open bisimilar in any context.
For the above example, this means that, because $\context{ A' }$ is not quasi-open bisimilar to $\context{ C' }$, we should also have that $A'$ is not quasi-open bisimilar to $C'$.
Indeed this property is satisfied by quasi-open bisimilarity.

Processes $A'$ and $C'$ are not quasi-open bisimilar.
Recall that a quasi-open bisimulation is preserved under fresh substitutions and $\sub{x}{\pk{k}}$ is fresh.
Thus we have the following.
\[
A'\sub{x}{\pk{k}} \lts{\co{a}(u)} {\sub{u}{r}} \cpar 0
\qquad
\mbox{and}
\qquad
C'\sub{x}{\pk{k}} \lts{\co{a}(u)} {\sub{u}{\cout{a}{\aenc{m}{\pk{k}}}}} \cpar 0
\]
The resulting processes are clearly not statically equivalent (consider $r = u$ under each substitutions above).
There is also a more subtle distinguishing strategy, we will return to in Section~\ref{section:modal}.

When processes are proven to be equivalent using quasi-open bisimilarity, compositional reasoning can be applied in confidence.
For example, standard rules expected in a \textit{structural congruence} hold according to quasi-open bisimilarity; hence can be safely applied anywhere in any process.
The following properties are useful in later in this work.
\begin{lemma}\label{lemma:algebra}
For all $P$, $Q$, $R$ and $S$, such that $x \not\in\fv{S}$, we have $0 \cpar P \sim P$ and, 
$\nu x.P \cpar S \sim \mathopen{\nu x.}\left(P \cpar S\right)$ only if , $\nu x.\nu y.P \sim \nu y.\nu x.P$ and $\nu x.0 \sim 0$,
$\left(P \cpar Q\right) \cpar R \sim P \cpar \left(Q \cpar R\right)$, and $P \cpar Q \sim Q \cpar P$.
\end{lemma}
\begin{proof}
Take the least open relation containing the bisimulation sets typically employed.
\end{proof}

Compositional reasoning can also be useful for reusability. Some properties may be verified on a sub-protocol, and, by compositionality, we can deduce they hold on a larger protocol.


\subsection{Quasi-open bisimilarity and open barbed bisimilarity coincide.}

As illustrated in the previous sub-section, a core guarantee offered by quasi-open bisimilarity is that it is a congruence relation.
In this section, we prove quasi-open bisimilarity is preserved by all contexts, notably under input prefixes; and, furthermore, coincides exactly with open barbed bisimilarity, which is the coarsest (strong) bisimilarity congruence.

We deliberately provide all important step of proofs in this section, to avoid uncertainty about this non-trivial result.
The novel cases for the following theorem are those showing quasi-open bisimilarity is preserved under mismatch, par and replication.
\begin{theorem}[contexts]\label{theorem:congruence}
If $P \sim Q$ then for all contexts $\context{\ \cdot\ }$, we have $\context{P} \sim \context{Q}$.
\end{theorem}
\begin{proof}
The proof can be broken into several lemmas, showing quasi-open bisimilarity is preserved under each process construct.
The most involved case, closure under parallel composition, is provided in Lemma~\ref{lemma:par}. Closure under replication is also quite involved hence provided in Lemma~\ref{lemma:rep}.
More immediate cases are presented below. In each of the following assume $P \sim Q$, and hence there exists quasi-open bisimulation such that $P \mathrel{\mathcal{R}} Q$.

Closure under input prefix is almost immediate.
Let $\mathcal{S}$ be the least \textit{open} relation (Def.~\ref{def:close}) extending $\mathcal{R}$ such that $M(x).P \mathrel{\mathcal{S}} M(x).Q$.
Now $M(x).P \lts{M\,N} P\sub{x}{N}$ and $M(x).Q \lts{M\,N} Q\sub{x}{N}$.
Furthermore, $P\sub{x}{N} \mathrel{\mathcal{S}} Q\sub{x}{N}$, since, by definition, $\mathcal{R}$ is an open relation.\footnote{Note closure under inputs does not hold for labelled bisimilarity, since in contrast to quasi-open bisimilarity, a labelled bisimulation is not necessarily preserved under substitutions. This is a key advantage of quasi-open bisimilarity.}

Closure under restriction is immediate, since, $\mathcal{R}$ is open, hence $\nu x.P \mathrel{\mathcal{R}} \nu x.Q$.

Closure under output, is also immediate.
Since $\mathcal{R}$ is an open relation, $\sub{u}{N} \cpar P \mathrel{\mathcal{R}} {\sub{u}{N}} \cpar Q$, for fresh $u$.
Also $\cout{M}{N}.P \lts{\co{M}(u)} {\sub{u}{N}} \cpar P$ and $\cout{M}{N}.Q \lts{\co{M}(u)} {\sub{u}{N}} \cpar Q$.
Hence relation $\mathcal{O}$, defined as the least open relation extending $\mathcal{R}$ such that $\cout{M}{N}.P \mathrel{\mathcal{O}} \cout{M}{N}.Q$ is a quasi-open bisimulation.

Closure under equality prefixes follows since $P \lts{\pi} A$ iff $\left(\match{M=N}P\right)\mathclose{\sigma} \lts{\pi\sigma} A\sigma$ for all $\sigma$ fresh for $\bn{\pi}$ such that $M\sigma = N\sigma$.
Since $P \mathrel{\mathcal{R}} Q$, there exists $B$ such that $Q \lts{\pi} B$ and $A \mathrel{\mathcal{R}} B$. Hence $\left(\match{M=N}Q\right)\mathclose{\sigma} \lts{\pi\sigma} B\sigma$ and, $\mathcal{R}$ is \textit{open} (Def.~\ref{def:close}), $\sigma$ is fresh for $\bn{\pi}$, and $\bn{\pi}$ is the domain of active substitutions in $A$ and $B$, we have $A\sigma \mathrel{\mathcal{R}} Q'\sigma$. Hence the least \textit{open} relation $\mathcal{M}$ extending $\mathcal{R}$ such that $\match{M=N}P \mathrel{\mathcal{M}} \match{M=N}Q$ is a quasi-open bisimulation.

Closure under mismatch is less obvious.
Now, assume $P \lts{\pi} A$, in which case, there exists $B$ such that $Q \lts{\pi} B$ and $A \mathrel{\mathcal{R}} B$.
Also, observe, by the \textsc{Mismatch} rule, 
$P \lts{\pi} A$ 
iff
 $\env \colon \left(\match{M\not=N}P\right)\mathclose{\sigma} \lts{\pi\sigma} A\sigma$ and
for all $\sigma$ fresh for $\bn{\pi}$ for all $\env$ such that $\env \vDash M\sigma \not= N\sigma$,
which holds iff 
$\mathopen{\nu \env.}\left( \sigma \cpar \left(\match{M\not=N}P\right)\mathclose{\sigma} \right) \lts{\pi} \mathopen{\nu \env.}\left( \sigma \cpar A\sigma \right)$.
Similarly, we have 
$Q \lts{\pi} B$ iff we have $\env \colon \left(\match{M\not=N}Q\right)\mathclose{\sigma} \lts{\pi\sigma} \mathopen{\nu \env.}\left( \sigma \cpar  B\sigma\right)$, for all $\sigma$ and $\env$ such that $\env \vDash M\sigma \not= N\sigma$.
Now consider the least \textit{open} relation $\mathcal{D}$ extending $\mathcal{R}$ such that $\match{M\not=N}P \mathrel{\mathcal{D}} \match{M\not=N}Q$.
Now, since $\mathcal{D}$ is \textit{open}, 
$\mathopen{\nu \env.}\left( \sigma \cpar \left(\match{M\not=N}P\right)\mathclose{\sigma} \right) \mathrel{\mathcal{D}} \mathopen{\nu \env.}\left( \sigma \cpar \left(\match{M\not=N}Q\right)\mathclose{\sigma} \right)$.\footnote{This is the raison d'{\^e}tre for extending environments in the definition reachability (Def.~\ref{definition:reachable}): to range over all extensions of environments and active substitutions enabling a mismatch.}
Also, $\mathcal{R}$ is \textit{open}, and $\sigma$ is fresh for $\bn{\pi}$ hence the domain of both $A$ and $B$, we have $\mathopen{\nu \env.}\left( \sigma \cpar A\mathclose{\sigma} \right) \mathrel{\mathcal{D}} \mathopen{\nu \env.}\left( \sigma \cpar B\mathclose{\sigma} \right)$.
Hence $\mathcal{D}$ is a quasi-open bisimulation, as required.

Closure under choice is standard. Take the least open relation $\mathcal{C}$ extending both $\mathcal{R}$ and the identity relation such that $P + R \mathrel{\mathcal{C}} Q + R$ and consider when $P + R \lts{\pi} A$.
Now, if $P \lts{\pi} A$, there exists $B$ such that $Q \lts{\pi} B$ and $A \mathrel{\mathcal{R}} B$; hence $Q + R \lts{\pi} Q'$ and $A \mathrel{\mathcal{C}} B$.
Otherwise $R \lts{\pi} A$; hence $Q + R \lts{\pi} A$ and $A \mathrel{\mathcal{C}} A$, as required.
\end{proof}

\begin{lemma}\label{lemma:par}
If $P \sim Q$, then $P \cpar R \sim Q \cpar R$.
\end{lemma}
\begin{proof}
Assume $P \sim Q$. Hence there exists quasi-open bisimulation $\mathcal{R}$ such that $P \mathrel{\mathcal{R}} Q$.
Now construct $\mathcal{S}$ to be the least \textbf{open} relation  (Def.~\ref{def:close}) such that
if
$\mathopen{\nu\vec{x}.}\left( \sigma \cpar  A\right) \mathrel{\mathcal{R}} \mathopen{\nu\vec{y}.}\left( \theta \cpar  B\right)$ and process $R$ is such that $\fv{R} \cap \left(\vec{x}\cup\vec{y}\right) = \emptyset$
then $\mathopen{\nu\vec{x}.}\left( A \cpar R\sigma \right) \mathrel{\mathcal{S}} \mathopen{\nu\vec{y}.}\left( B \cpar R\theta \right)$.
We aim to show $\mathcal{S}$ is a quasi-open bisimulation. Assume 
$\fv{R} \cap \left(\vec{x} \cup \vec{y}\right) = \emptyset$ and $\mathopen{\nu\vec{x}.}\left( \sigma \cpar  A\right) \mathrel{\mathcal{R}} \mathopen{\nu\vec{y}.}\left( \theta \cpar  B\right)$ in the following.
\begin{itemize}

\item Consider when $R \lts{\co{M}(u)} \mathopen{\nu \vec{z}.}\left( {\sub{u}{N}} \cpar S \right)$.
Since $\fv{R} \cap \left(\vec{x} \cup \vec{y} \right) = \emptyset$
it also holds that $\left(\fv{M} \cup \fv{S}\right) \cap \left(\vec{x} \cup \vec{y} \right) = \emptyset$; and $\vec{x} \colon R\sigma \lts{\co{M\sigma}(u)} \mathopen{\nu \vec{z}.}\left( {\sub{u}{N\sigma}} \cpar S\sigma \right)$, assuming without loss of generality that $\sigma$ is fresh for $u$ and $\vec{x}$. 
Now, assuming $\vec{z} \cap \left(\fv{P} \cup \fv{Q} \cup \vec{x} \cup \vec{y}\right) = \emptyset$, by the \textsc{Alias}, \textsc{Res} and \textsc{Par-r} rules,
we have transition $\vec{x} \colon \sigma \cpar P \cpar R\sigma \lts{\co{M}(u)} \mathopen{\nu \vec{z}.}\left( {\sub{u}{N}} \cpar \sigma \cpar P \cpar S \right)$.
Thereby, since $\fv{R} \cap \vec{x} = \emptyset$,
we have
$\mathopen{\nu \vec{x}.}\left( \sigma \cpar A \cpar R\sigma \right) \lts{\co{M}(u)} \mathopen{\nu \vec{x},\vec{z}.}\left( {\sub{u}{N}} \cpar \sigma \cpar A \cpar S \right)$.
Now, since $\mathopen{\nu\vec{x}.}\left( \sigma \cpar  A\right) \mathrel{\mathcal{R}} \mathopen{\nu\vec{y}.}\left( \theta \cpar  B\right)$,
and $\fv{S} \cap \left(\vec{x} \cup \vec{y} \right) = \emptyset$,
by definition of $\mathcal{S}$, we have $\mathopen{\nu\vec{x}.}\left( \sigma \cpar A \cpar S\sigma \right) \mathrel{\mathcal{S}} \mathopen{\nu\vec{y}.}\left( \theta \cpar  B \cpar S\theta \right)$;
and, since $\mathcal{S}$ is open (Def.~\ref{def:close}), we have the following.
\[
\mathopen{\nu\vec{x},\vec{z}.}\left( {\sub{x}{N\sigma}} \cpar \sigma \cpar A \cpar S\sigma \right) \mathrel{\mathcal{S}} \mathopen{\nu\vec{y},\vec{z}.}\left( {\sub{u}{N\theta}} \cpar \theta \cpar  B \cpar S\theta \right)
\]
Furthermore, by monotonicity, $\vec{x} \colon R\theta \lts{\co{M\theta}(u)} \mathopen{\nu \vec{z}.}\left( {\sub{u}{N\theta}} \cpar S\theta \right)$
and hence,
 by \textsc{Par-r}, \textsc{Alias}, and \textsc{Res} rules, $\mathopen{\nu\vec{y}.}\left( \theta \cpar  B \cpar R \right) \lts{M(u)} \mathopen{\nu\vec{y},\vec{z}.}\left( \sub{u}{N\theta} \cpar \theta \cpar  B \cpar S\theta \right)$, as required.

\item
Consider the case where $\vec{x} \colon A \lts{M\sigma\,N\sigma} A'$ and $R \lts{\co{M}(u)} \mathopen{\nu \vec{z}.}\left({\sub{u}{N}} \cpar S\right)$ and without loss of generality  $\vec{z} \cap \left( \fv{A} \cap \fv{B} \right) = \emptyset$.
We have
$\left(\fv{M} \cup \fv{N} \cup \fv{S}\right) \cap \left(\vec{x} \cup \vec{y}\right) = \emptyset$, since $\fv{R} \cap \left(\vec{x} \cup \vec{y}\right) = \emptyset$.
By monotonicity, we have $\vec{x} \colon R\sigma \lts{\co{M\sigma}(u)} \mathopen{\nu \vec{z}.}\left({\sub{u}{N\sigma}} \cpar S\sigma\right)$.
By rule \textsc{Close-r}, $\vec{x} \colon A \cpar R\sigma \lts{\tau} \mathopen{\nu \vec{z}.}\left(  A' \cpar S\sigma \right)$, since $\vec{z} \cap \left( \fv{A} \cap \fv{B} \right) = \emptyset$.
Thereby, by \textsc{Alias} and \textsc{Res}, we have $\mathopen{\nu\vec{x}.}\left( \sigma \cpar A \cpar R\sigma \right) \lts{\tau} \mathopen{\nu \vec{x}, \vec{z}.}\left( \sigma \cpar A' \cpar S\sigma \right)$.
Now, since we assumed $\mathopen{\nu\vec{x}.}\left( \sigma \cpar A \right) \mathrel{\mathcal{R}} \mathopen{\nu\vec{y}.}\left( \theta \cpar B \right)$,
and $\left(\fv{M} \cup \fv{N}\right) \cap \vec{y} = \emptyset$
there exists $B'$ such that
$\mathopen{\nu\vec{y}.}\left( \theta \cpar B \right) \lts{M\,N} \mathopen{\nu\vec{y}.}\left( \theta \cpar B' \right)$ and $\mathopen{\nu\vec{x}.}\left( \sigma \cpar A'\right) \mathrel{\mathcal{R}} \mathopen{\nu\vec{y}.}\left( \theta \cpar B'\right)$. 
Hence it must be the case that $\vec{y} \colon B \lts{M\theta\,N\theta} B'$.
By monotonicity we have $\vec{y} \colon R\theta \lts{\co{M\theta}(u)} \mathopen{\nu \vec{z}.}\left({\sub{u}{N\theta}} \cpar S\right)$;
hence, by rule $\textsc{Close-l}$ we have $\vec{y} \colon B \cpar R\theta \lts{\tau} \mathopen{\nu \vec{z}.}\left(  B' \cpar S\theta \right)$, so
by rules \textsc{Alias} and \textsc{Res}, we have $\mathopen{\nu \vec{y}.}\left( \theta \cpar B \cpar R\theta\right) \lts{\tau} \mathopen{\nu \vec{y},\vec{z}.}\left( \theta \cpar B' \cpar S\theta \right)$.
Since $\fv{S} \cap \left(\vec{x} \cup \vec{y}\right) = \emptyset$, by definition of $\mathcal{S}$ we have 
$\mathopen{\nu \vec{x},\vec{z}.}\left( \sigma \cpar A' \cpar S\sigma \right) \mathrel{\mathcal{S}} \mathopen{\nu \vec{y},\vec{z}.}\left( \theta \cpar B' \cpar S\theta \right)$, as required.

\item
Consider the case where $\vec{x} \colon A \lts{\co{M\sigma}(u)} \mathopen{\nu \vec{v}.}\left({\sub{u}{K}} \cpar A'\right)$ and $R \lts{M\,u} S$, for fresh $u$.
Since $\left(\vec{x} \cup \vec{y}\right) \cap \fv{R} = \emptyset$, we have $\left(\vec{x} \cup \vec{y}\right) \cap \left(\fv{M} \cup \fv{S}\right) = \emptyset$.
By monotonicity and freshness of $u$, we have $\vec{x} \colon R\sigma \lts{M\sigma\,K} S\sigma\sub{u}{K}$; hence $\vec{x} \colon A \cpar R\sigma \lts{\tau} \mathopen{\nu \vec{v}.}\left( A' \cpar S\sigma\sub{u}{K} \right)$, by rule \textsc{Close-l};
and, furthermore,  
$\mathopen{\nu \vec{x}.}\left( \sigma \cpar A \cpar R\sigma \right) \lts{\tau} \mathopen{\nu \vec{x},\vec{v}.}\left( \sigma \cpar A' \cpar S\sigma\sub{u}{K} \right)$,
by rules \textsc{Res} and \textsc{Alias}.
Now, since $\vec{x} \cap \fv{M} = \emptyset$,
we have $\mathopen{\nu\vec{x}.}\left( \sigma \cpar A \right) \lts{\co{M}(u)} \mathopen{\nu \vec{x}, \vec{v}.}\left(\sigma \cpar  {\sub{u}{K}} \cpar A'\right)$, by rules \textsc{Res} and \textsc{Alias}.
Thereby, $\mathopen{\nu\vec{x}.}\left( \sigma \cpar  A\right) \mathrel{\mathcal{R}} \mathopen{\nu\vec{y}.}\left( \theta \cpar  B\right)$
and $\mathcal{R}$ is a quasi-open bisimulation,
there exists $B'$, $\vec{w}$ and $L$ such that
$\mathopen{\nu\vec{y}.}\left( \theta \cpar B\right) \lts{\co{M}(u)} \mathopen{\nu \vec{y}, \vec{w}.}\left(\sigma \cpar {\sub{u}{L}} \cpar B'\right)$
and
$\mathopen{\nu \vec{x}, \vec{v}.}\left(\sigma \cpar  {\sub{u}{K}} \cpar A'\right) \mathrel{\mathcal{R}} \mathopen{\nu \vec{y}, \vec{w}.}\left(\sigma \cpar {\sub{u}{L}} \cpar B'\right)$.
Now make two observations. Firstly, by unfolding rules we have 
$\vec{y} \colon B \lts{\co{M\theta}(u)} \mathopen{\nu \vec{w}.}\left({\sub{u}{L}} \cpar B'\right)$;
and, by monotonicity and freshness of $u$, we have $\vec{y} \colon R\theta \lts{M\theta\,L} S\theta\sub{u}{L}$.
Hence, by rule \textsc{Close-l}, we have that $\vec{y} \colon B \cpar R\theta \lts{\tau} \mathopen{\nu \vec{w}.}\left( B' \cpar S\theta\sub{u}{L} \right)$;
and so
$\mathopen{\nu \vec{y}.}\left( \theta \cpar B \cpar R\theta \right) \lts{\tau} \mathopen{\nu \vec{y}, \vec{w}.}\left( \theta \cpar B' \cpar S\theta\sub{u}{L} \right)$,
by rules \textsc{Res} and \textsc{Alias}.
Secondly, since $\fv{S} \cap \left(\vec{x} \cup \vec{y} \right) = \emptyset$, by definition of $\mathcal{S}$, 
we have 
$\mathopen{\nu \vec{x}, \vec{v}.}\left( \sigma \cpar A' \cpar S\sigma\sub{u}{K} \right) \mathrel{\mathcal{S}} \mathopen{\nu \vec{y},\vec{w}.}\left(\theta \cpar B' \cpar S\theta\sub{u}{L} \right)$,
as required.

\item Remaining cases, where $A$ or $R$ act independently, are similar to the first case above.
\end{itemize}
Thereby $\mathcal{S}$ is a quasi-open bisimulation such that $P \cpar R \mathrel{\mathcal{S}} Q \cpar R$; hence $P \cpar R \sim Q \cpar R$.
\end{proof}

\begin{lemma}\label{lemma:rep}
If $P \sim Q$, then $\bang P \sim \bang Q$.
\end{lemma}
\begin{proof}
Assume $P \sim Q$, hence there exists a quasi-open bisimulation $\mathcal{R}$ such that $P \mathrel{\mathcal{R}} Q$.
Now define $\mathcal{S}_0$ to be the singleton relation such that $\bang P \mathrel{\mathcal{S}_0} \bang Q$.
Inductively, define $\mathcal{S}_{n+1}$ to be the least relation such that if
$\mathopen{\nu\vec{x}.}\left(\sigma \cpar A\right) \mathrel{\mathcal{R}} \mathopen{\nu\vec{y}.}\left(\theta \cpar B\right)$
and
$\mathopen{\nu\vec{v}.}\left(\varsigma \cpar C\right) \mathrel{\mathcal{S}_{n}} \mathopen{\nu\vec{w}.}\left(\vartheta \cpar D\right)$
then we have
$\mathopen{\nu\vec{x},\vec{v}.}\left(A\varsigma \cpar C\sigma\right) \mathrel{\mathcal{S}_{n+1}} \mathopen{\nu\vec{y},\vec{w}.}\left(B\vartheta \cpar D\theta\right)$.
Define $\mathcal{S}$ to be the least open relation containing $\bigcup_{n\in\omega} \mathcal{S}_n$.  
Now, assume
$\mathopen{\nu \vec{x}.}\left( \sigma \cpar P_1 \cpar \hdots P_n \cpar \bang P\right) \mathrel{\mathcal{S}_n} \mathopen{\nu \vec{y}.}\left( \theta \cpar Q_1 \cpar \hdots Q_n \cpar \bang Q\right)$
and consider the following.
\begin{itemize}
\item 
Assume $P \lts{\co{M}(u)} \mathopen{\nu \vec{v}.}\left( \sub{u}{N} \cpar P'\right)$ and $P \lts{M\,u} R$, for fresh $u$ and $\fv{P} \cap \vec{v} = \emptyset$.
By monotonicity, $P \lts{M\,N} R\sub{u}{N}$.
By \textsc{Rep-close}, $\bang P \lts{\tau} \mathopen{\nu \vec{v}.}\left(P' \cpar R\sub{u}{N} \cpar \bang P\right)$.
Now since $P \mathrel{\mathcal{R}} Q$,
there exists $\vec{w}$, $L$ and $Q'$ such that $Q \lts{\co{M}(u)} \mathopen{\nu \vec{w}.}\left( \sub{u}{L} \cpar Q'\right)$
and also we have 
$\mathopen{\nu \vec{v}.}\left( \sub{u}{N} \cpar P'\right) \mathrel{\mathcal{R}} \mathopen{\nu \vec{w}.}\left( \sub{u}{L} \cpar Q'\right)$.
Furthermore, there exists $S$ such that $Q \lts{M\,u} S$ and $R \mathrel{\mathcal{R}} S$;
hence, by monotonicity, $Q \lts{M\,L} S\sub{u}{L}$.
Without loss of generality we can assume $\vec{w} \cap \fv{Q} = \emptyset$;
thus, by \textsc{Rep-close}, $\bang Q \lts{\tau} \mathopen{\nu \vec{w}.}\left( Q' \cpar S\sub{u}{L} \cpar \bang Q\right)$.
Furthermore, by definition of $\mathcal{S}_{2}$, we have
$\mathopen{\nu \vec{v}.}\left(P' \cpar R\sub{u}{N} \cpar \bang P\right) \mathrel{\mathcal{S}_{2}} \mathopen{\nu \vec{w}.}\left( Q' \cpar S\sub{u}{L} \cpar \bang Q\right)$.
Extending inductively, over the definition of $\mathcal{S}_n$, we have the following, as required.
\[
\mathopen{\nu \vec{x},\vec{v}.}\left( \sigma \cpar P_1 \cpar \hdots P_n \cpar P' \cpar R\sub{u}{N} \cpar \bang P\right) \mathrel{\mathcal{S}_{n+2}} \mathopen{\nu \vec{y},\vec{w}.}\left( \theta \cpar Q_1 \cpar \hdots Q_n \cpar Q' \cpar S\sub{u}{L} \cpar \bang Q\right)
\]

\item There are several more cases to consider, where in each case an action on the left of $\mathcal{S}_n$ can be matched by an action on the right, such that the resulting processes stay within $\mathcal{S}$.
\begin{itemize}
\item Some $P_i \lts{\pi} P'_i$ acts independently, staying within $\mathrel{\mathcal{S}_n}$. 
\item For $i \not= j$, $P_i \lts{\co{M}(u)} A_i$ and $P_j \lts{M\,N} P'_j$, resulting a $\tau$ transition, staying within $\mathrel{\mathcal{S}_n}$.
\item $P \lts{\pi} P'$ acts independently, applying rule \textsc{Rep-act}, progressing to $\mathrel{\mathcal{S}_{n+1}}$.
\item $P_i \lts{\co{M}(u)} A_i$ and $P \lts{M\,N} P'$, resulting in a $\tau$ transition progressing to $\mathrel{\mathcal{S}_{n+1}}$.
\item $P \lts{\co{M}(u)} A$ and $P_j \lts{M\,N} P_j'$, resulting in a $\tau$ transition progressing to $\mathrel{\mathcal{S}_{n+1}}$.
\end{itemize}
The proofs for these cases do not differ significantly from what is already presented for parallel composition and \textsc{Rep-close}, hence are ommitted.
\end{itemize}
Thereby $\mathcal{S}$ is a quasi-open bisimulation such that $\bang P \mathrel{\mathcal{S}} \bang Q$; hence $\bang P \sim \bang Q$.
\end{proof}

Given Theorem~\ref{theorem:congruence},
the soundness of quasi-open bisimilarity with respect to open barbed bisimilarity is standard.
For a self-contained presentation, we recall the proof.
\begin{corollary}[soundness]\label{theorem:soundness}
If $P \sim Q$ then $P \bsim Q$.
\end{corollary}
\begin{proof}
Assume $P \sim Q$.
Symmetry follows immediately from the definition.
By Theorem~\ref{theorem:congruence}, $\context{P} \sim \context{Q}$.
By closure under transitions, if $P \lts{\tau} P'$ then there exists $Q'$ such that $Q \lts{\tau} Q'$ and $P' \sim Q'$.
If $\barb{P}{M}$ then there exists $A$ such that $P \lts{\co{M}(x)} A$ or $P \lts{M\,N} A$.
In the former case, there exists $B$ such that $Q \lts{\co{M}(x)} B$, 
similarly, in the latter case, there exists $B$ such that $Q \lts{M\,N} B$. Hence in either case $\barb{Q}{M}$.
Hence $\sim$ is an open barbed bisimulation; thus $P \bsim Q$, as required.
\end{proof}

For completeness, we require that open barbed bisimilarity is preserved under any substitution.
\begin{lemma}\label{lemma:sub}
If $P \bsim Q$, then, for any substitution $\sigma$, $P\sigma \bsim Q\sigma$.
\end{lemma}
\begin{proof}
Assume $P \bsim Q$ and consider substitution $\sigma$ defined such that $\sub{z_1, \hdots, z_n}{K_1, \hdots, K_n}$.
Since open barbed bisimilarity is preserved under all contexts,
for fresh names $c$,
we have the following.
\[
{\cout{c}{K_1}.\cout{c}{K_2}\hdots \cout{c}{K_n}} \cpar c(z_1).c(z_2)\hdots c(z_n).P \bsim {\cout{c}{K_1}.\cout{c}{K_2}\hdots \cout{c}{K_n}} \cpar c(z_1).c(z_2)\hdots c(z_n).Q
\]
Each of these processes can perform the same number of $\tau$-transitions to reach the states $0 \cpar P\sigma$ and $0 \cpar Q\sigma$
Since open barbed bisimilarity is closed under $\tau$-transitions, we have $0 \cpar P\sigma \bsim 0 \cpar Q\sigma$.
By Lemma~\ref{lemma:algebra}, $0 \cpar R \sim R$ holds, and, by Corollary~\ref{theorem:soundness}, $0 \cpar R \bsim R$; hence $P\sigma \bsim Q\sigma$ as required.
\end{proof}


The following result supports our claim that our definition of quasi-open bisimilarity for the applied $\pi$-calculus, Definition~\ref{definition:quasi-open}, is correct, and a canonical choice of (strong interleaving) bisimilarity.
Recall open barbed bisimilarity has an objective language-independent definition.

\begin{theorem}[completeness]\label{theorem:openbarbed}
Quasi-open bisimilarity coincides with open barbed bisimilarity.
\end{theorem}
\begin{proof}
Define relation $\mathrel{\mathcal{R}}$
such that
$\mathopen{\nu \vec{y}.}\left( \sigma \cpar P \right)
\mathrel{\mathcal{R}} 
\mathopen{\nu \vec{z}.}\left( \rho \cpar Q \right)$,
where $\sigma = \sub{x_1, \hdots x_n}{M_1, \hdots M_n}$ and
$\rho = \sub{x_1, \hdots x_n}{N_1, \hdots N_n}$,
whenever
for some fresh names $\vec{a} = \left\{ a_1, a_2, \hdots, a_n \right\}$ and $I = \left\{1, \hdots n\right\}$
we have
 $P_1 \bsim Q_1$,
such that
$P_1 \triangleq \mathopen{\nu \vec{y}.}\left(\prod_{i\in I}{\bang \cout{a_i}{M_i}} \cpar P \right)$
and
$Q_1 \triangleq \mathopen{\nu \vec{z}.}\left(\prod_{i\in I}{\bang \cout{a_i}{N_i}} \cpar Q \right)$.

Note $\prod_{i\in I}{S_i}$ is an abbreviation for $S_1 \cpar \hdots \cpar S_n$.

The symmetry of $\mathcal{R}$ is immediate from the symmetry of $\bsim$.
The rest of the proof shows $\mathcal{R}$ is a quasi-open bisimulation.
In each clause assume
$\mathopen{\nu \vec{y}.}\left( \sigma \cpar P \right)
\mathrel{\mathcal{R}} 
\mathopen{\nu \vec{z}.}\left( \rho \cpar Q \right)$.
Hence, by definition, we have $P_1 \bsim Q_1$, as define above.

\textbf{Static equivalence.}
Consider the following context, where $s$ is a fresh name.
\[
\context{\ \cdot\ } \triangleq a_1(x_1).a_2(x_2).\hdots \mathopen{a_n(x_n).}{\match{M = N}{\cout{s}{s}}} \cpar \left\{\ \cdot\ \right\}
\]
Also assume $\left(\vec{y}\cup \vec{z}\right) \cap \left(\fv{M} \cup \fv{N}\right) = \emptyset$.
Since open barbed bisimilarity is closed under all contexts, $\context{P_1} \bsim \context{Q_1}$.
By closure under $\tau$-transitions the following are open barbed bisimilar.
\[
\mathopen{\nu \vec{y}.}\left({\match{M\sigma = N\sigma}{\cout{s}{s}}} \cpar \prod_{i\in I}{\bang \cout{a_i}{M_i}} \cpar P \right)
\bsim
\mathopen{\nu \vec{z}.}\left({\match{M\rho = N\rho}{\cout{s}{s}}} \cpar \prod_{i\in I}{\bang \cout{a_i}{N_i}} \cpar Q \right)
\]
By definition of open barbed bisimilarity if either of the above processes exhibits barb $s$ then so must the other. By unfolding the rules for labelled transitions, it must be the case that $M\sigma = N\sigma$ if and only if $M\rho = N\rho$.
Thus for any choice of $M$ and $N$ such that $\left(\vec{y}\cup \vec{z}\right) \cap \left(\fv{M} \cup \fv{N}\right) = \emptyset$, we have $M\sigma = N\sigma$ if and only if $M\rho = N\rho$.
Thus 
$\mathopen{\nu \vec{y}.}\left( \sigma \cpar P \right)$
and
$\mathopen{\nu \vec{z}.}\left( \rho \cpar Q \right)$
are statically equivalent.

\textbf{Closure under input transitions.}
Assume $\mathopen{\nu \vec{y}.}\left( \sigma \cpar P \right) \lts{K\,L} A$, such that, without loss of generality, $\left(\fv{K}\cup\fv{L}\right) \cap \left(\vec{y} \cup \vec{z}\right) = \emptyset$.
Consider the following context, where $s$ is a fresh name.
\[
\context{\ \cdot\ } \triangleq a_1(x_1).a_2(x_2).\hdots \mathopen{a_n(x_n).}\left({\cout{s}{s}} + {\cout{K}{L}}\right) \cpar \left\{\ \cdot\ \right\}
\]
Since open barbed bisimilarity is closed under all contexts,
$\context{P_1} \bsim \context{Q_1}$ holds.
Both processes have only one way to perform $n$ transitions to reach a pair of states exhibiting barb $s$.
Thus we have the following, by closure of an open barbed bisimulation under $\tau$-transitions.
\[
\mathopen{\nu \vec{y}.}\left( \left({\cout{s}{s}} + {\cout{K\sigma}{L\sigma}} \right) \cpar \prod_{i\in I}{\bang \cout{a_i}{M_i}} \cpar P \right)
\quad\bsim\quad
\mathopen{\nu \vec{z}.}\left( \left({\cout{s}{s}} + {\cout{K\rho}{L\rho}} \right) \cpar \prod_{i\in I}{\bang \cout{a_i}{N_i}} \cpar Q \right)
\]

Now, since $\mathopen{\nu \vec{y}.}\left( \sigma \cpar P \right) \lts{K\,L} A$, it must the the case that $A = \mathopen{\nu \vec{y}.}\left( \sigma \cpar P' \right)$
 for some $P'$.
Hence, by unfolding the definition of labelled transitions, we have $\vec{y} \colon P \lts{K\sigma\,L\sigma} P'$.
Now by the \textsc{Out} and \textsc{Sum-r} rules, $\vec{y} \colon {{\cout{s}{s}}} + {\cout{K\sigma}{L\sigma}} \lts{\co{K\sigma}(w)} {\sub{w}{L\sigma}} \cpar 0$ for fresh $w$.
Hence, by rules \textsc{Par-l}, \textsc{Close-l} and \textsc{Res}, we can construct the following transition.
\[
\mathopen{\nu \vec{y}.}\left( \left({{\cout{s}{s}}} + {{\cout{K\sigma}{L\sigma}}} \right) \cpar \prod_{i\in I}{\bang \cout{a_i}{M_i}} \cpar P \right)
\lts{\tau}
\mathopen{\nu \vec{y}.}\left( 0 \cpar \prod_{i\in I}{\bang \cout{a_i}{M_i}} \cpar P' \right)
\]
Notice the above transition reaches a state where there is no barb $s$. By the definition of an open barbed bisimulation, there exists a transition of the following form, where $\barb{R}{s}$ does not hold.
\[
\mathopen{\nu \vec{z}.}\left( \left({\cout{s}{s}} + {\cout{K\rho}{L\rho}} \right) \cpar \prod_{i\in I}{\bang \cout{a_i}{N_i}} \cpar Q \right) \lts{\tau} R
\]
Furthermore it must be the case that $\mathopen{\nu \vec{y}.}\left( 0 \cpar \prod_{i\in I}{\bang \cout{a_i}{M_i}} \cpar P' \right) \bsim R$.
Since $R$ does not have barb $s$, and $s$ was chosen fresh, hence there is no input on channel $s$, the following output must have been performed: $\vec{z} \colon {{\cout{s}{s}}} + {\cout{K\rho}{L\rho}} \lts{\co{K\rho}(w)} {\sub{w}{L\rho}} \cpar 0$ . Thus there exists $Q'$ such that $\vec{z} \colon Q \lts{K\rho\,L\rho} Q'$ and hence $R = \mathopen{\nu \vec{z}.}\left( 0 \cpar \prod_{i\in I}{\bang \cout{a_i}{N_i}} \cpar Q' \right)$.

From the above observations, and since $\left(\fv{K}\cup\fv{L}\right) \cap \vec{z} = \emptyset$, we can construct transition $\mathopen{\nu \vec{z}.}\left( \rho \cpar Q \right) \lts{K\,L} \mathopen{\nu \vec{z}.}\left( \rho \cpar Q' \right)$.
By Lemma~\ref{lemma:algebra}, $0 \cpar S \sim S$, hence, by Corollary~\ref{theorem:soundness}, $0 \cpar S \bsim S$, for all processes. Thereby we have 
$\mathopen{\nu \vec{y}.}\left( \prod_{i\in I}{\bang \cout{a_i}{M_i}} \cpar P' \right) \bsim \mathopen{\nu \vec{z}.}\left( \prod_{i\in I}{\bang \cout{a_i}{N_i}} \cpar Q' \right)$.
So, by definition of $\mathcal{R}$, we have 
$\mathopen{\nu \vec{y}.}\left( \sigma \cpar P' \right) \mathrel{\mathcal{R}} \mathopen{\nu \vec{z}.}\left( \rho \cpar Q' \right)$, as required.

\textbf{Closure under output transitions.}
Assume $\mathopen{\nu \vec{y}.}\left( \sigma \cpar P \right) \lts{K(u)} \mathopen{\nu \vec{y},\vec{w}.}\left( \sub{u}{L} \cpar \sigma \cpar P' \right)$, where, without loss of generality, $\fv{K} \cap \left(\vec{x} \cup \vec{y}\right) = \emptyset$. 
Consider the following context for fresh $s$ and $b$.
\[
\context{\ \cdot\ } \triangleq a_1(x_1).a_2(x_2).\hdots \mathopen{a_n(x_n).}\left({\cout{s}{s}} + {\cin{K}{\var}.\bang\cout{b}{\var}}\right) \cpar \left\{\ \cdot\ \right\}
\]
Since open barbed bisimilarity is closed under all contexts,
$\context{P_1} \bsim \context{Q_1}$ holds.
Both processes can perform $\tau$-transitions $m$ times, to reach the following pair of states exhibiting barb $s$.
\[
\mathopen{\nu \vec{y}.}\left( \left({\cout{s}{s}} + {\cin{K\sigma}{\var}.\bang\cout{b}{\var}}\right) \cpar \prod_{i\in I}{\bang \cout{a_i}{M_i}} \cpar P \right)
\bsim
\mathopen{\nu \vec{z}.}\left( \left({\cout{s}{s}} + {\cin{K\rho}{\var}.\bang\cout{b}{\var}}\right) \cpar \prod_{i\in I}{\bang \cout{a_i}{N_i}} \cpar Q \right)
\]
The above are open barbed bisimilar, since open barbed bisimilarity is preserved under $\tau$-transitions.

Now, by unfolding the definition of labelled transitions, we have $\vec{y} \colon P \lts{K\sigma(u)} \mathopen{\nu \vec{v}.}\left({\sub{u}{L}} \cpar P' \right)$ and $\vec{v} \cap \fv{K} = \emptyset$. Also, by rules \textsc{Inp} and \textsc{Sum-r}, we have $\vec{y} \colon \left({\cout{s}{s}} + {\cin{K\sigma}{\var}.\bang\cout{b}{\var}}\right) \lts{{K\sigma}\,L\sigma} {\bang\cout{b}{L\sigma}}$.
Hence, by rules \textsc{Par-l}, \textsc{Close-r} and \textsc{Res} we have the following interaction.
\[
\mathopen{\nu \vec{y}.}\left( \left({\cout{s}{s}} + {\cin{K\sigma}{\var}.\bang\cout{b}{\var}}\right) \cpar \prod_{i\in I}{\bang \cout{a_i}{M_i}} \cpar P \right)
\lts{\tau}
\mathopen{\nu \vec{y}, \vec{v}.}\left( {\bang\cout{b}{L\sigma}} \cpar \prod_{i\in I}{\bang \cout{a_i}{M_i}} \cpar P' \right)
\]
Notice the above transition reaches a state with  barb $b$. By the definition of an open barbed bisimulation, there exists a transition of the following form.
\[
\mathopen{\nu \vec{z}.}\left( \left({\cout{s}{s}} + {\cin{K\rho}{\var}.\bang\cout{b}{\var}}\right) \cpar \prod_{i\in I}{\bang \cout{a_i}{N_i}} \cpar Q \right) \lts{\tau} R
\]
Furthermore, it must be the case that $\mathopen{\nu \vec{x}.}\left( 0 \cpar \prod_{i\in I}{\bang \cout{a_i}{M_i}} \cpar P' \right) \bsim R$ where $\barb{R}{b}$ holds.
Since $\barb{R}{b}$ and $b$ was chosen fresh, input transition 
$\vec{z} \colon {\cout{s}{s}} + {\cin{K\rho}{\var}.\bang\cout{b}{\var}} \lts{K\rho\,L'} {\bang\cout{b}{L'}}$
must be triggered, for some $L'$ (not necessarily equivalent to $L$).
Thus we must have that $\vec{z} \colon Q \lts{K\rho(u)} \mathopen{\nu \vec{w}.}\left( \sub{u}{L'} \cpar Q' \right)$ and $\vec{w} \cap \fv{K} = \emptyset$;
and hence, by rules \textsc{Par-l}, \textsc{Close-r} and \textsc{Res}, $R = \mathopen{\nu \vec{z}, \vec{w}.}\left( {\bang\cout{b}{L'}} \cpar \prod_{i\in I}{\bang \cout{a_i}{N_i}} \cpar Q' \right)$.

From the above, we can construct transition $\mathopen{\nu \vec{z}.}\left( \rho \cpar Q \right) \lts{K(u)} \mathopen{\nu \vec{z}, \vec{w}.}\left( \sub{u}{L'} \cpar \rho \cpar Q' \right)$.
Since 
$\mathopen{\nu \vec{y},\vec{v}.}\left( {\bang\cout{b}{L\sigma}} \cpar \prod_{i\in I}{\bang \cout{a_i}{M_i}} \cpar P' \right) \bsim \mathopen{\nu \vec{z}, \vec{w}.}\left( {\bang\cout{b}{L'}} \cpar \prod_{i\in I}{\bang \cout{a_i}{N_i}} \cpar Q' \right)$, by definition of $\mathcal{R}$, we have 
$\mathopen{\nu \vec{y},\vec{v}.}\left( \sigma \cpar \sub{u}{L} \cpar P' \right) \mathrel{\mathcal{R}} \mathopen{\nu \vec{z},\vec{w}.}\left( \rho' \cpar \sub{u}{L'} \cpar Q' \right)$, as required.

\textbf{Closure under $\tau$ transitions.} Assume $\mathopen{\nu \vec{y}.}\left( \sigma \cpar P \right) \lts{\tau} \mathopen{\nu \vec{y}.}\left( \sigma \cpar P'\right)$.
Unfolding rules labelled transitions, $\vec{y} \colon P \lts{\tau} P'$; hence
$
\mathopen{\nu \vec{y}.}\left( \prod_{i\in I}\bang \cout{a_i}{M_i\sigma} \cpar P \right)
\lts{\tau} 
\mathopen{\nu \vec{y}.}\left( \prod_{i\in I}\bang \cout{a_i}{M_i\sigma} \cpar P' \right)$.
Hence, since open barbed bisimulations are closed under $\tau$-transitions, for some $R$, we have
$Q_1 \lts{\tau} R$ such that $\mathopen{\nu \vec{y}.}\left( \prod_{i\in I}\bang \cout{a_i}{M_i\sigma} \cpar P' \right) \bsim R$.
Since $a_i$ are fresh they cannot be involved in $\tau$-transitions, so by unfolding the rules of the labelled transition system, for some $Q'$, we have $R = \mathopen{\nu \vec{z}.}\left(\prod_{i\in I}{\bang \cout{a_i}{N_i}} \cpar Q' \right)$
and $Q \lts{\tau} Q'$.
From the above we can construct $\tau$-transition $\mathopen{\nu \vec{z}.}\left( \rho \cpar Q \right) \lts{\tau} \mathopen{\nu \vec{z}.}\left( \rho \cpar Q' \right)$.
Furthermore, by definition of $\mathcal{R}$, we have 
$
\mathopen{\nu \vec{y}.}\left( \sigma \cpar P' \right)
\mathrel{\mathcal{R}}
\mathopen{\nu \vec{z}.}\left( \rho \cpar Q' \right)
$, as required.

\textbf{Closure under reachability.}
Assume idempotent substitutions $\theta$ and $\vartheta$ are fresh for $\vec{x}$.
To avoid clashes between names $\vec{a}$ and $\theta$ and $\vartheta$, select fresh names $\left\{ b_i \right\}_{i \in \left\{1,\hdots,n\right\}}$.
Let substitution $\theta'$, with domain $\vec{a}$ be such that $a_i\theta' = b_i$.
So, by Lemma~\ref{lemma:sub}, $P_1\theta' \bsim Q_1\theta'$,
hence, since $a_i$ are fresh for $M_i$ and $P$, we have
$\mathopen{\nu \vec{y}.}\left(\prod_{i\in I}{\bang \cout{b_i}{M_i}} \cpar P \right) \bsim \mathopen{\nu \vec{z}.}\left(\prod_{i\in I}{\bang \cout{b_i}{N_i}} \cpar Q \right)$.
Without loss of generality, applying $\alpha$-conversion,
assume $\theta$ and $\vartheta$ are fresh for $\vec{y}\cup\vec{z}$;
hence, by Lemma~\ref{lemma:sub},
$\mathopen{\nu \vec{y}.}\left(\prod_{i\in I}{\bang \cout{b_i}{M_i\theta\vartheta}} \cpar P\theta\vartheta \right)
\bsim
\mathopen{\nu \vec{z}.}\left(\prod_{i\in I}{\bang \cout{b_i}{N_i\theta\vartheta}} \cpar Q\theta\vartheta \right)$.
Now define $\sigma' = \sub{x_1, \hdots x_n}{M_1\theta\vartheta, \hdots M_n\theta\vartheta}$ and $\rho' = \sub{x_1, \hdots x_n}{N_1\theta\vartheta, \hdots N_n\theta\vartheta}$.
Since $\sigma$ and $\rho$ are idempotent and $\theta$ and $\vartheta$ are fresh for $\vec{x}$; $\sigma'$ and $\rho'$ are idempotent.
Finally, consider context $\mathopen{\nu \vec{w}.}\left({\prod_{j\in J}\bang\cout{c_i}{K_i} } \cpar \left\{\ \cdot\ \right\}\right)$, for $\vartheta = \sub{v_1,\hdots,v_m}{K_1,\hdots,K_m}$, $J = \left\{ 1, \hdots, m\right\}$ and fresh $\left\{c_1, \hdots, c_m\right\}$.
Note, by Lemma~\ref{lemma:algebra}, $R \cpar \mathopen{\nu x.}S \sim \mathopen{\nu x.}\left(R \cpar S\right)$, for $x \not\in\fv{R}$, hence, by Corollary~\ref{theorem:soundness}, 
$R \cpar \mathopen{\nu x.}S \bsim \vec{x}.\left(R \cpar S\right)$.
Observe that, by closure of open barbed bisimilarity under contexts, and the aforementioned scope extrusion property we have the following.
\[
\mathopen{\nu \vec{w},\vec{y}.}\left({\prod_{j\in J}\bang\cout{c_i}{K_i} } \cpar \prod_{i\in I}{\bang \cout{b_i}{M_i\theta\vartheta}} \cpar P\theta\vartheta \right) \bsim \mathopen{\nu \vec{w},\vec{z}.}\left({\prod_{j\in J}\bang\cout{c_i}{K_i} } \cpar \prod_{i\in I}{\bang \cout{b_i}{N_i\theta\vartheta}} \cpar Q\theta\vartheta \right)
\]
Hence, by definition of $\mathcal{R}$, we have
$\mathopen{\nu \vec{w},\vec{y}.}\left( \vartheta' \cpar \sigma' \cpar P\theta\vartheta \right) \mathbin{\mathcal{R}} \mathopen{\nu \vec{w},\vec{z}.}\left( \vartheta' \cpar \rho' \cpar Q\theta\vartheta \right)$, as required.

Thus the relation $\mathcal{R}$ is a quasi-open barbed bisimulation. Furthermore if $P \bsim Q$ then $P \mathrel{\mathcal{R}} Q$. Thereby $P \bsim Q$ implies $P\sim Q$ (the converse to Corollary~\ref{theorem:soundness}).
\end{proof}

It is interesting to compare the above proof to the corresponding proof for the $\pi$-calculus~\cite{Sangiorgi2001}.
In the corresponding proof for the $\pi$-calculus checks are built into bound output transitions to ensure extruded private names are fresh.
In the above proof no such checks are required for output transitions; such checks are subsumed by checking static equivalence.

\section{Characterising open barbed bisimilarity for the applied $\pi$-calculus using an intuitionistic modal logic}\label{section:modal}

A modal logic characterises a bisimilarity whenever bisimilar processes satisfy the same formulae~\cite{hennessy85jacma}.
Recent insight~\cite{Ahn2017}, has shown that intuitionistic modal logics can be used to characterise bisimilarity congruences.
In this section, we consider how the modal logic called intuitionistic $\FM$~\cite{Horne2018}, characterising open barbed bisimilarity
lifts to the setting of the applied $\pi$-calculus.

A syntax for $\FM$ is presented below.
\begin{gather*}
\begin{array}{l}
\left.
\begin{array}{rlr}
\phi \Coloneqq& \ttt & \mbox{top} \\
          \mid& \fff & \mbox{bottom} \\
          \mid& M = N & \mbox{equality} \\
          \mid& \phi \wedge \phi & \mbox{conjunction} \\
          \mid& \phi \vee \phi & \mbox{disjunction} \\
          \mid& \phi \yields \phi & \mbox{implication} \\
\end{array}
\right\}
\mbox{intuitionistic logic}
\\
\left.
\begin{array}{rlr}
\hspace{17pt}
          \mid& \diam{\pi}\phi & \hspace{21pt} \mbox{diamond} \\
          \mid& \boxm{\pi}\phi & \mbox{box}
\end{array}
\right\}
\mbox{modalities}
\end{array}
\begin{array}{l}
\mbox{common abbreviations:}
\\[4pt]
\neg \phi \triangleq \phi \yields \fff
\\
M \not= N \triangleq \neg (M = N)
\end{array}
\end{gather*}
In the syntax above, observe connectives cover the standard conjunction, disjunction, implication, top and bottom of intuitionistic logic with equalities. 
The two modalities box and diamond range over all observable actions.
Observable actions $\pi$, as defined in Sec.~\ref{section:api-open}, range over $\tau$, bound outputs and free inputs. 
Intuitionistic negation is defined as a standard abbreviation.

\begin{figure}[h]
\begin{gather*}
\begin{array}{lcl}
A \vDash \ttt &&\mbox{always holds.} \\
\mathopen{\nu \vec{x}.}\left( \sigma \cpar P \right) \vDash M = N
&\mbox{iff}& M\sigma \mathrel{=_E} N\sigma  ~\mbox{and}~ \vec{x}\cap\left(\fv{M}\cup\fv{N}\right) = \emptyset
\\
A \vDash \phi_1 \land \phi_2 &\mbox{iff}&
  A \vDash \phi_1 ~\mbox{and}~ A \vDash \phi_2.
\\
A \vDash \phi_1 \lor \phi_1 &\mbox{iff}&
  A \vDash \phi_1 ~\mbox{or}~ A \vDash \phi_2.
\\
A \vDash \phi_1 \yields \phi_2 &\mbox{iff}&
\mbox{whenever } \reachable{A}{\sigma,\vec{x}.\rho}{A'},  ~\mbox{we have}~
A' \vDash \phi_1\sigma ~\mbox{implies}~ A' \vDash \phi_2\sigma.
\\
A \vDash \diam{\pi}\phi &\mbox{iff}&
  \mbox{there exists }B \mbox{ such that } A \lts{\pi} B ~\mbox{and}~ B \vDash \phi.
\\
A \vDash \boxm{\pi}\phi &\mbox{iff}&
\mbox{whenever } \reachable{A}{\sigma,\vec{x}.\rho}{A'}, \mbox{ and } 
    A' \lts{\pi\sigma} B,  ~\mbox{we have}~
    B \vDash \phi\sigma.
\end{array}
\end{gather*}
\vspace*{-.75em}
\caption{The semantics of intuitionistic modal logic $\FM$, adapted for the applied $\pi$-calculus.}\label{figure:FM}
\end{figure}

There are two differences between 
intuitionistic $\FM$ for the applied $\pi$-calculus, presented in Fig.~\ref{figure:FM}, and previous work on intuitionistic $\FM$ for the $\pi$-calculus with mismatch.
Firstly, any messages, not just variables, can appear in equalities and on labels.
Secondly, the definition is surprisingly more concise: there is no free output modality; and extruded private names are recorded in the extended processes so need not be accounted for in the satisfaction relation.

Soundness of quasi-open bisimilarity with respect to intuitionistic $\FM$ is established by a straightforward induction of the structure of formulae.
\begin{theorem}[soundness]\label{theorem:sound-modal}
If $P \sim Q$, then for all $\phi$,
$P \vDash \phi$ if and only if $Q \vDash \phi$.
\end{theorem}

For what follows we restrict to finitary message theories.
\begin{definition}
An equational theory is \textit{finitary} whenever, 
for all messages $M$ and $N$, there is a finite set of substitutions $\left\{ \sigma_i \right\}_{i \in I}$ 
such that, for all $i\in I$, we have $M\sigma_i =_E N\sigma_i$ and, for all $\theta$ such that $M\theta =_E N\theta$, there exists $j \in I$ such that $\sigma_j \leq \theta$ (i.e., for some $\varsigma$, we have $\sigma_j\cdot\varsigma = \theta$).
\end{definition}
For example Dolev-Yao, and our example theory in Fig.~\ref{figure:messages} are finitary.
Message theories with an associative operator, such a string concatenation, are not finitary in general.
However, theories with an associative-commutative operator~\cite{Ayala-Rincon2017} are finitary.

The following contrapositive to completeness holds
under certain assumptions sufficient to ensure a finite formula can be constructed.
\begin{theorem}[distinguishing formulae]\label{theorem:complete-modal}
For fragments of the applied $\pi$-calculus that are decidable with a finitary equational theory,
if $P \not\sim Q$, there exists $\phi_L$ such that $P \vDash \phi_L$ and $Q \not\vDash \phi_L$, and also 
there exists $\phi_R$ such that $Q \vDash \phi_R$ and $P \not\vDash \phi_R$.
\end{theorem}
The proof is similar to the proof for the $\pi$-calculus with mismatch except that there is an additional reason processes may be distinguished in the non-bisimulation strategy --- namely the processes are not statically equivalent.


Note Theorem~\ref{theorem:complete-modal} may hold under weaker conditions, lifting the restriction that we consider only fragments where quasi-open bisimilarity is decidable.
However, the above result is still useful, adequate for a large class of useful theories and processes.
The proof yields an algorithm for generating distinguishing formulae from distinguishing strategies obtained from where the search for a bisimulation fails.


\subsection{Examples of distinguishing formulae expressed using intuitionistic $\FM$.}
We present examples illustrating subtleties of the logic and also provide distinguishing formulae for examples discussed previously.

\paragraph{Subtle formulae requiring absence of law of excluded middle.}
In Section~\ref{section:compositionality}, we presented a distinguishing strategy for the following processes.
\[
A' \triangleq {\cout{a}{r}}
\qquad
\mbox{v.s.}
\qquad
C' \triangleq 
\texttt{if}\,x = \pk{k}\,\texttt{then}\,\cout{a}{\aenc{m}{\pk{k}}}\,\texttt{else}\,\cout{a}{r}
\]
A more subtle distinguishing strategy than that presented in Section~\ref{section:compositionality} also exists.
The more subtle strategy exploits the absence of the law of excluded middle as follows.
Observe $A' \lts{\co{a}(u)} {\sub{u}{r}} \cpar 0$, cannot be matched by any transition of $C'$ without additional assumptions about $x$ and $k$.
Thus for a distinguishing formula biased to the left we have the following.
\[
 A' \vDash \diam{a(u)}\ttt
\qquad
\mbox{and}
\qquad
 C' \not\vDash \diam{a(u)}\ttt
\]
Now observe that $C'$ can only perform an output transition either: under substitutions $\sigma$ such that $x\sigma =_E \pk{k}\mathclose{\sigma}$; or,
under substitutions $\rho$ and environments $\vec{n}$ such that $\vec{n} \vDash x\rho \not= \pk{k}\mathclose{\rho}$.
For a distinguishing formula biased to the right we write a box modality followed by the strongest post-condition after an output is performed, i.e., either $x = \pk{k}$ or $x \not= \pk{k}$, as follows.
\[
 A' \not\vDash \boxm{a(u)}\left( x = \pk{k} \vee x \not= \pk{k} \right)
\qquad
\mbox{and}
\qquad
 C' \vDash \boxm{a(u)}\left( x = \pk{k} \vee x \not= \pk{k} \right)
\]
Observe that in a classical setting neither of the above formulae would be distinguishing.
In a classical modal logic we have $\boxm{a(u)}\left( x = \pk{k} \vee x \not= \pk{k} \right)$ is a tautology, due to the law of excluded middle.
Thus the absence of the law of excluded middle for intuitionistic $\FM$ provides additional distinguishing power.

\paragraph{Is absence of law of excluded middle necessary?}
For the above example there are distinguishing formulae where the absence of the law of excluded middle is not necessary for the formula to be distinguishing.
For example, we have the following distinguishing formula biased to the right, which would also be distinguishing in a classical variant of $\FM$.
\[
A' \not\vDash x = \pk{k} \yields \diam{\co{a}(u)}\left( u = \aenc{m}{\pk{k}} \right)
\qquad
C' \vDash x = \pk{k} \yields \diam{\co{a}(u)}\left( u = \aenc{m}{\pk{k}} \right)
\]
There are however examples for which the intuitionistic nature of $\FM$ is necessary for a distinguishing formula to exist.
In the \textit{classical} setting of labelled bisimilarity the following are equivalent.
\[
\mbox{let}~
D' \triangleq 
\cout{a}{\aenc{\pair{m}{r}}{\pk{k}}}
~\mbox{in}
\quad
R \triangleq a(x).C' + a(x).A' + a(x).D'
\quad
\mbox{v.s.}
\quad
S \triangleq a(x).A' + a(x).D'
\]
Intuitively, both processes $R$ and $S$ above model servers that either behave as $A'$ or $D'$ regardless of the input.
In addition, process $R$ above also has the option of receiving an input then making a decision based on the input whether to behave as $A'$ or as $D'$.
Classically, $R$ and $S$ are equivalent, since upon receiving an input corresponding to the prefix of $a(x).C'$ we have decided immediately that either $x = \pk{k}$ or $x \not= \pk{k}$ holds; hence the appropriate branch in the first process is taken.

However, in the \textit{intuitionistic} setting, as required for quasi-open bisimilarity, the above processes are distinguished. The distinguishing strategy is as follows: $a(x).C' + a(x).A' + a(x).D' \lts{a\,x} C'$ can only be matched by either $a(x).A' + a(x).D' \lts{a\,x} A'$ or $a(x).A' + a(x).D' \lts{a\,x} D'$, and neither $C' \sim A'$ nor $C' \sim D'$ hold.
This strategy leads to the following distinguishing formulae in intuitionistic $\FM$ biased to process $R$ and $S$ respectively.
\[
R \vDash \diam{a\,x}\boxm{a(u)}\left( x = \pk{k} \vee x \not= \pk{k} \right)
\qquad\mbox{and}
\qquad
S \vDash \boxm{a\,x}\diam{a(u)}\ttt
\]
The \textit{a posteriori} reason for not assuming the law of excluded middle everywhere is that this assumption is necessary to characterise a bisimilarity congruence.
An \textit{a priori} justification for examples such as the above is less obvious.
We attempt an explanation as follows. In the above example, variable $x$ on the input label is not a ground term hence has not been fully read. Thus the program can proceed without fully reading $x$, lazily reading any sub-term only when required (in this case, to determine whether or not $x = \pk{k}$).
By analogy, when you download this paper you do not read every character before proceeding with your next task, but in the future you may study a detail of Theorem~\ref{theorem:openbarbed} (efficient rather than lazy would be a better term).

In contrast, if we ground the names as follows the following are bisimilar, even in the intuitionistic setting of quasi-open bisimilarity.
\[
\mathopen{\nu k.\cout{a}{\pk{k}}.}\left( a(x).C' + a(x).A' + a(x).D' \right)
\quad\sim\quad
\mathopen{\nu k.\cout{a}{\pk{k}}.}\left( a(x).A' + a(x).D' \right)
\]

\paragraph{Static equivalence examples.}

The processes $\nu m, n. {\cout{a}{m}.\cout{a}{n}} \mathrel{\not\lsim} \nu n.\cout{a}{n}.\cout{a}{\hash{n}}$ are not open bisimilar.
A distinguishing strategy is that both processes can perform output transitions
$\co{a}(u)$ and $\co{a}(v)$ reaching the pair of processes 
$\mathopen{\nu m, n.}\left( \sub{u,v}{m,n} \cpar 0 \right) \not\lsim \mathopen{\nu n.}\left(\sub{u,v}{n,\hash{n}} \cpar 0\right)$. As discussed in Section~\ref{section:static}, these processes are not statically equivalent; with distinguishing messages $v$ and $\hash{u}$.
Thereby we can construct the following distinguishing formulae biased to each respective process.
\[
\nu m, n. \cout{a}{m}.\cout{a}{n} \vDash \diam{\co{a}(u)}\diam{\co{a}(v)}\left(v \not= \hash{u}\right)
\quad
\nu n.\cout{a}{n}.\cout{a}{h(n)} \vDash \boxm{\co{a}(u)}\boxm{\co{a}(v)}\left(v = \hash{u}\right)
\]

Recall from Section~\ref{section:api-open}, we have $\nu x.{\cout{a}{\aenc{x}{z}}} \mathrel{\not\lsim} \nu x.\cout{a}{\aenc{\pair{x}{y}}{z}}$.
The distinguishing strategy, described previously, involved substitution $\sub{z}{\pk{w}}$, and distinguishing recipes $\snd{\adec{u}{w}}$ and $y$, each of which are recorded in the following distinguishing formula biased to the right.
\[
\nu x.\cout{a}{\aenc{\pair{x}{y}}{z}} \vDash \diam{\co{a}(u)}\left(z = \pk{w} \yields \snd{\adec{u}{w}} = y\right)
\]
From the same strategy, we can construct the following distinguishing formula biased to the left.
\[
\nu x.\cout{a}{\aenc{x}{z}} \vDash \boxm{\co{a}(u)}\left(\snd{\adec{u}{w}} \not= y\right)
\]
Sub-formula $z = \pk{w} \yields \adec{u}{k} = x \yields x = y$ requires explanation. After applying substitution $\sub{u}{\aenc{y}{z}}\cdot \sub{z}{\pk{k}}$ to $\adec{u}{k} =_E x$, resulting messages $\adec{\aenc{y}{\pk{k}}}{k}$ and $x$ are not related by the equational theory. However, there is a strongest postcondition $x=y$ enabling this equation since $\adec{\aenc{y}{\pk{k}}}{k} =_E y$. We write the strongest postcondition after the distinguishing recipe $\adec{u}{k} = x$.

\paragraph{Mobility example with messages as channels.}
For the mobility example,
$\nu m.\cout{a}{\pair{m}{n}}.m(x)
\not\lsim
\nu m.\cout{a}{\pair{m}{n}}$ from Section~\ref{section:api-open},
the distinguishing strategy presented previously yields the following formula biased to the respective processes.
\[
\nu m.\cout{a}{\pair{m}{n}}.m(x)
\vDash
\diam{\co{a}(u)}\diam{\fst{u}\,x}\ttt
\qquad
\nu m.\cout{a}{\pair{m}{n}}
\vDash
\boxm{\co{a}(u)}\boxm{\fst{u}\,x}\fff
\]
Notice that message $\fst{u}$ is used as a channel name for the second output action.

\section{Examples of the Theory Applied to Security and Privacy Properties}
\label{section:privacy}

We illustrate here the power of the theory developed on two more substantial examples.
The first is an established privacy example, demanding mismatch.
The second is an example involving a larger message theory (blind signatures). Neither scenario could previously be verified using a bisimilarity congruence in the literature, such as open bisimilarity for the spi-calculus without mismatch.

\subsection{Privacy property for which mismatch is necessary.}

We provide a more elaborate version of the running example of a private server, adapted from the literature~\cite{Abadi2004,Cheval2017}.
In this protocol, there are two servers: the first responds in a way only the owner of private key $a$ can detect; while the second only responds to the owner of private key $b$.
The aim of the protocol is to ensure that an external observer cannot determine the intended recipient of data.
The scenario can be modelled by the following processes.
\[
P \triangleq
\match{\snd{\adec{y}{c}} = \pk{a}}
\nu n.\cout{x}{\aenc{\pair{\fst{\adec{y}{c}}}{\pair{n}{\pk{c}}}}{\pk{a}}}
\]
\[
\context{\ \cdot\ } \triangleq \nu a,b,c.\cout{x}{\pk{a}}.\cout{x}{\pk{b}}.\cout{x}{\pk{c}}.\mathclose{\left\{\ \cdot\ \right\}}
\]
\[
Broken\_Server \triangleq
\begin{array}[t]{l}
\context{x(y).P}
\end{array}
\qquad
\mbox{and}
\qquad
Broken\_Server' \triangleq
\begin{array}[t]{l}
\context{x(y).P\sub{a}{b}}
\end{array}
\]
Notice $Broken\_Server$ and $Broken\_Server'$ differ only by the name $a$ or $b$ respectively in sub-process $P$.
The former, after the input action, will only respond to a message containing public key $\pk{a}$ with a message readable by the owner of secret key $a$; while the latter only responds to a message containing $\pk{b}$, producing a message readable by the owner of secret key $b$.

$Broken\_Server$ and $Broken\_Server'$ are not quasi-open bisimilar,
as witnessed by the following distinguishing formula biased to $Broken\_Server$.
\[
Broken\_Server \vDash \diam{\co{x}(u)}\diam{\co{x}(v)}\diam{\co{x}(w)}\diam{x\,\aenc{\pair{z}{u}}{w}}\diam{\co{x}(s)}\ttt
\]
\[
Broken\_Server' \not\vDash \diam{\co{x}(u)}\diam{\co{x}(v)}\diam{\co{x}(w)}\diam{x\,\aenc{\pair{z}{u}}{w}}\diam{\co{x}(s)}\ttt
\]
This indicates that the protocol does not preserve the privacy of the client to whom the server uniquely responds. By following the distinguishing strategy given by the formula, an attacker can distinguish the server responding to $\pk{a}$ from the server responding to $\pk{b}$.

The privacy of the above protocol can be fixed by inserting dummy messages, so the attacker cannot distinguish between a response intended for the owner of secret key $a$ and a response intended for the owner of $b$.
To prove this privacy property, we establish the following two processes are quasi-open bisimilar,
where $P$ is as defined above.
\[
Fixed\_Server \triangleq
\context{
\mathopen{x(y).}\left(P + \match{\snd{\adec{y}{c}} \not= \pk{a}}\nu m. \cout{x}{m} \right)
}
\]
\[
Fixed\_Server' \triangleq
\context{
\mathopen{x(y).}\left(P\sub{a}{b} + \match{\snd{\adec{y}{c}} \not= \pk{b}}\nu m. \cout{x}{m} \right)
}
\]
After three sends,
$\co{x}(u)$, $\co{x}(v)$, and $\co{x}(w)$,
we reach the following extended processes.
\[
A \triangleq 
\mathopen{\nu a,b,c.}
\left(
\sub{u,v,w}{\pk{a},\pk{b},\pk{c}}
\cpar
\mathopen{x(y).}\left(P + \match{\snd{\adec{y}{c}} \not= \pk{a}}\nu m. \cout{x}{m} \right)
\right)
\]
\[
B \triangleq 
\mathopen{\nu a,b,c.}
\left(
\sub{u,v,w}{\pk{a},\pk{b},\pk{c}}
\cpar
\mathopen{x(y).}\left(P\sub{a}{b} + \match{\snd{\adec{y}{c}} \not= \pk{b}}\nu m. \cout{x}{m} \right)
\right)
\]
There are three cases to consider at this point triggering different behaviours.
\begin{itemize}
\item Message $\aenc{\pair{z}{u}}{w}$ is input on channel $x$. Since $u$ is an alias for $\pk{a}$ this represents trying to trigger the server to respond to the owner of secret key $a$ (as defined inside $P$).

\item Message $\aenc{\pair{z}{v}}{w}$ is input on channel $x$. Since $v$ is an alias for $\pk{b}$ this represents trying to trigger the server to respond to the owner of secret key $b$ (as in $P\sub{a}{b}$).

\item Any other message is input on channel $x$. 
\end{itemize}

In the first case above, after input $\aenc{\pair{z}{u}}{w}$ 
the match guard in sub-process $P$ is triggered.
This is because after the input the match guard is $\snd{\adec{\aenc{\pair{z}{u}}{w}}{c}} = \pk{a}$ under substitution $\sub{u,v,w}{\pk{a},\pk{b},\pk{c}}$, amounting to 
$\snd{\adec{\aenc{\pair{z}{\pk{a}}}{\pk{c}}}{c}} =_E \pk{a}$, which holds.
For $B$, the mismatch guard is triggered, 
since $a,b,c \vDash \pk{a} \not= \pk{b}$ holds. 
Thereby, both $A$ and $B$ can perform an output $\co{x}(t)$ to reach the following extended processes.
\[
A \lts{x\,\aenc{\pair{z}{u}}{w}}
\lts{\co{x}(t)}
\mathopen{\nu a,b,c,n.}
\left(
\sub{u,v,w,t}{\pk{a},\pk{b},\pk{c},\aenc{\pair{z}{\pair{n}{\pk{c}}}}{\pk{a}}}
\cpar
0
\right)
\]
\[
B \lts{x\,\aenc{\pair{z}{u}}{w}}
\lts{\co{x}(t)}
\mathopen{\nu a,b,c,m.}
\left(
\sub{u,v,w,t}{\pk{a},\pk{b},\pk{c},m}
\cpar
0
\right)
\]
The above extended processes are statically equivalent, even
under all substitutions fresh for $\left\{u,v,w,a,b,c\right\}$, as required for an open relation.
The private key $a$ and nonce $n$ are never revealed; hence $\aenc{\pair{z}{\pair{n}{\pk{c}}}}{\pk{a}}$ cannot be decrypted or reconstructed by an attacker, and thereby cannot be distinguished from a random cyphertext represented by $m$.

The second case is symmetric to the first case. Simply swap $a$ and $b$ in the argument above.

In the third case any other input, say $M$,
triggers the mismatch in both branches,
since, for any messages $M$ other that those equivalent to $\pk{a}$ or $\pk{b}$, we have 
$a,b,c \vDash M \not= \pk{a}$
and
$u,v,w \vDash M \not= \pk{b}$.
Thereby there is a quasi-open bisimulation $\mathcal{R}$ such that $A \mathrel{\mathcal{R}} B$, from which $Fixed\_Server \sim Fixed\_Server'$ follows.

Future work will explain larger privacy examples involving mismatch, that can be analysed using quasi-open bisimilarity. For example, the established attack on unlinkability of the French e-passport~\cite{Arapinis2010} can be discovered. Note, contrary to claims in that paper, there is an attack on the UK e-passport that is discovered quickly using quasi-open bisimilarity, and is confirmed by equivalence checking tools based on trace equivalence~\cite{Cheval2018}.


\subsection{Example using an extended message theory, featuring blind signatures.}

Now extend the message theory. Extend messages with signatures $\sign{M}{N}$.
Firstly, observe the following are quasi-open bisimilar.
\[
R \triangleq \nu n.\cout{a}{n}.a(x).\nu k.\cout{a}{\sign{x}{k}}.a(y).\match{
  y = \sign{n}{k}
}\tau
\]\[
S \triangleq\nu n.\cout{a}{n}.a(x).\nu k.\cout{a}{\sign{x}{k}}.a(y).\match{
  y = \sign{n}{k}
}\match{x = n}\tau
\]
Exploiting the compositionality of quasi-open bisimilarity, it is sufficient to check the following sub-processes are quasi-open bisimilar.
\[
\nu k.\cout{a}{\sign{x}{k}}.a(y).\match{
 y = \sign{n}{k}
}\tau
\]\[
\nu k.\cout{a}{\sign{x}{k}}.a(y).\match{
  y = \sign{n}{k}
}\match{x = n}\tau
\]
Since quasi-open bisimilarity is closed under context $\mathopen{\nu n.\cout{a}{n}.a(x).}\left\{\ \cdot\ \right\}$, it follows that $R \sim S$.

Now, extend further the example message language to model blind signatures~\cite{Chaum1983}. We have \textit{blinding} $\blind{M}{N}$, \textit{unblinding} $\unblind{M}{N}$ and the following additional equation.
\[
 \unblind{\sign{\blind{M}{N}}{K}}{N} = \sign{M}{K}
\]
Under this extended theory with blind signatures the above processes $R$ and $S$ are no longer quasi-open bisimilar.
A distinguishing strategy is explained by the following distinguishing formula biased to process $R$.
\[
R
\vDash
\diam{\co{a}(u)}\diam{a\,\blind{u}{z}}\diam{\co{a}(v)}\diam{a\,\unblind{v}{z}}\diam{\tau}\ttt
\]
The above formula describes strategies by which the above processes can be distinguished under a blind signature theory.
To see why, observe after four transitions we reach the following extended process, which is expected to satisfy formula $\diam{\tau}\ttt$.
\[\small
\mathopen{\nu k, n.}\left( \sub{u,v}{n,\sign{\blind{u}{z}}{k}} \cpar \match{
  \unblind{\sign{\blind{n}{z}}{k}}{z} = \sign{n}{k}
}\tau \right)
\]
The guard is satisfied by using the equation for blind signatures,
thus the $\tau$-transition is enabled, as required.
The above example, illustrates a \textit{signature forgery attack}, where two distinct signatures, valid for distinct messages are produced from a single signature.



An approach to avoiding signature forgery is to enforce a hash-and-sign approach to signatures.
This is illustrated by the following example where a hash function is inserted.
The following processes are quasi-open bisimilar, hence the first process cannot pass the test of providing two distinct messages signed with key $k$.
\[
\begin{array}{l}
\nu k. a(x).\cout{a}{\sign{x}{k}}.a(y).a(z).\match{y = \sign{\hash{m}}{k}}\match{z = \sign{\hash{n}}{k}}\match{m \not= n}\tau
\\
\qquad\qquad\qquad\qquad\qquad\qquad\qquad\qquad\qquad\qquad\qquad\qquad\quad
\lsim
\nu k. a(x).\cout{a}{\sign{x}{k}}.a(y).a(z)
\end{array}
\]
The key observation for verifying the above equivalence is, after four actions $a\,x$, $\co{a}(u)$, $a\,y$, $a\,z$ we reach the following extended process.
\[
\mathopen{\nu k.}
\left(
\sub{u}{\sign{x}{k}}
\cpar
\match{y = \sign{\hash{m}}{k}}\match{z = \sign{\hash{n}}{k}}\match{m \not= n}\tau
\right)
\quad
\mbox{v.s.}
\quad
\mathopen{\nu k.}
\left(
\sub{u}{\sign{x}{k}}
\cpar
0
\right)
\]
Regardless of what values for $x$, $y$ and $z$ there is no permitted choice such that all match and mismatch guards can be satisfied.
Any solution for passing the match guards forces $m = n$ to hold; hence the $\tau$ transition cannot be enabled.
Note this can be verified by solving a system of what are known as \textit{deducibility constraints}\footnote{Deducibility constraints $\vdash x$ and $\sign{x}{k} \vdash \sign{\hash{m}}{k}$ and $\sign{x}{k} \vdash \sign{\hash{n}}{k}$, where $k$ is a private name.} generated from the order messages are communicated, as explained in related work~\cite{Bursuc2014}.


\section{Comparison to Related Work on Labelled Bisimilarity}\label{section:related}



The bisimilarities presented so far are strong; in contrast to weak semantics that allows internal $\tau$-transitions to be ignored.
All previous work on bisimilarity for the applied $\pi$-calculus concerns weak semantics, mainly due to the hybrid reduction/label transition style used in the original paper, that does not permit a strong semantics to be expressed.

To define a weak semantics modify the labelled transitions in Fig.~\ref{figure:active} such that rules for choice, match and mismatch are replaced by the following direct definition of \texttt{if-then-else}.
\[
\begin{prooftree}
M =_E N
\justifies
\env \colon \texttt{if}\,M = N\,\texttt{then}\,P\,\texttt{else}\,Q
\lts{\tau} P
\end{prooftree}
\qquad
\qquad
\begin{prooftree}
\env \vDash M \not= N
\justifies
\env \colon \texttt{if}\,M = N\,\texttt{then}\,P\,\texttt{else}\,Q
\lts{\tau} Q
\end{prooftree}
\]
For the weak variant of open barbed bisimilarity, each $\tau$-transition can be matched by zero or more $\tau$-transitions;
as represented by $Q \Lts{} Q'$ in the following definition.
\begin{definition}[weak open barbed bisimilarity]\label{def:weak}
Define $\Barb{R}{M}$ whenever there exists $R'$ such that $R \Lts{} R'$ and $\barb{R'}{M}$.
A weak open barbed bisimulation $\rel$ is a symmetric relation over processes
such that whenever $P \rel Q$ holds 
the following hold:
\begin{itemize}
\item For all contexts $\context{\ \cdot\ }$, $\context{P} \rel \context{Q}$.
\item If $\barb{P}{M}$ then $\Barb{Q}{M}$.
\item If $P \lts{\tau} P'$, there exists $Q'$ such that $Q \Lts{} Q'$ and $P' \rel Q'$ holds.
\end{itemize}
Weak open barbed bisimilarity $\approxeq$ is the greatest weak open barbed bisimulation.
\end{definition}
Weak quasi-open bisimilarity differs from quasi-open bisimilarity with respect to the rules for \texttt{if-then-else},
and in the use of weak transitions. In the following $B \Lts{\pi} B'$ permits zero or more $\tau$-transitions before and after action $\pi$, and in the case $\pi = \tau$, $B \Lts{\pi} B'$ is $B \Lts{} B'$, defined above, thereby permitting one $\tau$-transition to be matched by zero transitions.
\begin{definition}[labelled bisimilarity]\label{definition:labelled}
A symmetric relation between extended processes $\mathrel{\mathcal{R}}$ is a labelled bisimulation whenever,
if $A \mathrel{\mathcal{R}} B$ then the following hold:
\begin{itemize}
\item $A$ and $B$ are statically equivalent.
\item If $A \lts{\pi} A'$ there exists $B'$ such that $B \Lts{\pi} B'$ and $A' \mathrel{\mathcal{R}} B'$.
\end{itemize}
Weak quasi-open bisimilarity $\approx$ is the greatest \textbf{open} labelled bisimulation (Def.~\ref{def:close}).
Labelled bisimilarity $\sim_\ell$ is the greatest labelled bisimulation, where the inequality between messages in the \texttt{else} branch is evaluated classically rather than intuitionistically.
\end{definition}
Theorem~\ref{theorem:openbarbed} extends to the weak case: \textit{weak open barbed bisimilarity} coincides with \textit{weak quasi-open bisimilarity}.
Changes to the proof to handle a weak semantics are minimal (a few extra observables are required in the proof of completeness). 
Note the above rules for \texttt{if-then-else} avoid well known compositionality problems with respect to choice in the weak setting.
Since $0 \mathrel{\approxeq} \tau.0$, and weak open barbed bisimilarity is a congruence, we have $\texttt{if}\,M = N\,\texttt{then}\,P\,\texttt{else}\,0 \mathrel{\approxeq} \texttt{if}\,M = N\,\texttt{then}\,P\,\texttt{else}\,\tau.0$.
However, this congruence property would fail for the weak variant of quasi-open bisimilarity, if, instead of the rules for \texttt{if-then-else} above, ``reactive'' choice in Fig.~\ref{figure:active} was employed.
Consequently, the above rules are essential for the soundness of weak quasi-open bisimilarity with respect to weak open barbed bisimilarity.

Weak quasi-open bisimilarity is clearly sound with respect to \textit{labelled bisimilarity}.
The only difference between labelled bisimilarity and weak quasi-open bisimilarity is the keyword \textit{open} from the definition above, and the use of classical negation when interpreting guards. Thereby labelled bisimilarity equates strictly more processes than weak quasi-open bisimilarity.
Indeed, in the classical setting of labelled bisimilarity, all free variables in processes are treated as ground terms; which makes redundant the set of names in the environment of the labelled transition system.
Thus the only difference compared to weak quasi-open bisimulation is that labelled bisimilarity is not preserved under \textit{reachability} (Def.~\ref{definition:reachable}).

\textit{Labelled bisimilarity} has been proven to coincide with \textit{observational equivalence}~\cite{Abadi16}.
Observational equivalence is a constrained version of Definition~\ref{def:weak} where contexts are restricted to ``evaluation contexts'', essentially of the form $\left\{\ \cdot\ \right\} \cpar P$.
Hence, by definition, weak open barbed bisimilarity is contained in observational equivalence.
Thereby, as an immediate consequence of the coincidence of labelled bisimilarity and observational equivalence, we have that weak open barbed bisimilarity is (strictly) finer than labelled bisimilarity.
\begin{corollary}[soundness w.r.t.\ labelled bisimilarity]
If $P \approxeq Q$, then $P \mathrel{\sim_\ell} Q$.
\end{corollary}
Of course, by the weak variant of Theorem~\ref{theorem:openbarbed}, this means any property proven using weak quasi-open bisimilarity is also valid for labelled bisimilarity.
Most notions of bisimilarity previously introduced for cryptographic calculi coincide with observational equivalence~\cite{abadi98njc,abadi01popl,boreale02sjc,borgstrom05mscs,Bengtson2009,Borgstrom2009,Liu2010,Johansson2010,Johansson2012,Abadi16}.
Intermediate results on symbolic bisimulations~\cite{borgstrom04concur,Delaune2007} also closely approximate observational equivalence.
Thus this work, respects the aims of all such papers, without loss of power for capturing security and privacy properties; while, in addition, providing the benefits of a bisimilarity congruence.

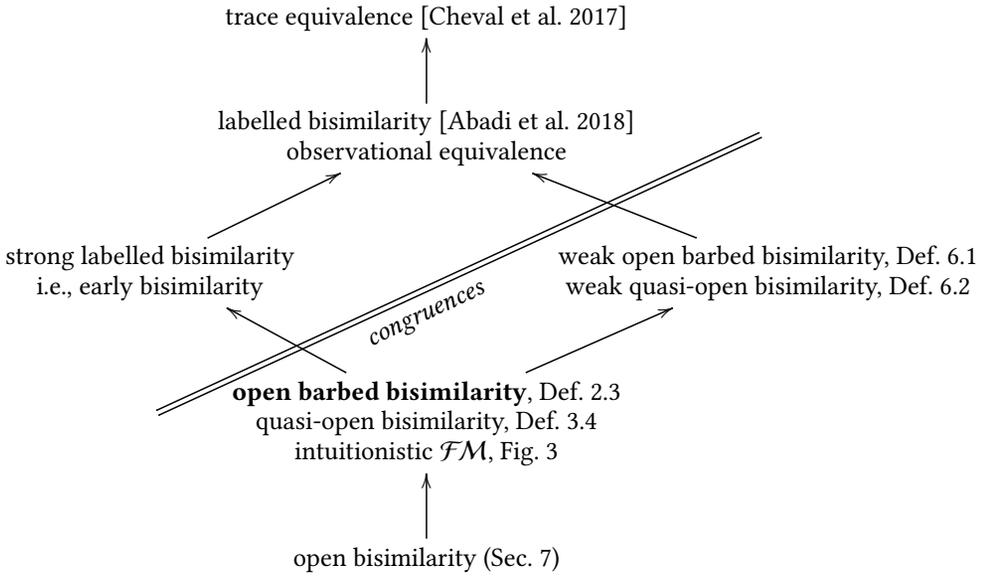
\begin{figure}
\xymatrix@C=-35pt{
& \txt{trace equivalence~\cite{Cheval2017}} \\
& \ar[u]\txt{labelled bisimilarity~\cite{Abadi16} \\ observational equivalence} & \\
\ar[ur]\txt{strong labelled bisimilarity \\
i.e.,\ early bisimilarity
}
&&
\ar[ul]\txt{weak open barbed bisimilarity, Def.~\ref{def:weak} \\ weak quasi-open bisimilarity, Def.~\ref{definition:labelled}}
\\
\ar@{=}^{\raisebox{-30pt}[-100pt][-100pt]{\rotatebox{24.8}{\textit{congruences}}}}[uurr]
&\ar[ul]\ar[ur] \txt{\textbf{open barbed bisimilarity}, Def.~\ref{def:open-barbed}\\
quasi-open bisimilarity, Def.~\ref{definition:quasi-open}
\\ intuitionistic $\FM$, Fig.~\ref{figure:FM} 
} \\
& \ar[u]\txt{open bisimilarity (Sec.~\ref{section:open})} 
&
}
\caption{Part of the spectrum of bisimilarities surrounding open barbed bisimilarity.}\label{fig:spectrum}
\end{figure}

The relations described in this subsection are summarised in Fig.~\ref{fig:spectrum}.
Observe the weak and strong variants of open barbed bisimilarity, below the double line are congruences; in contrast to labelled bisimilarity.
Note strong labelled bisimilarity is obtained from Def.~\ref{definition:quasi-open} by removing keyword \textit{open}.

Related work~\cite{Huttel2007} logically characterises observational equivalence using a classical modal logic.
The classical modal logic provided in that work is quite different from the classical variant of $\FM$, obtained by removing the requirement that implication and box are preserved under reachability.
Other work~\cite{Parrow2017}, introduces a more abstract classical modal logic characterising weak bisimilarities, that may be instantiated for the applied $\pi$-calculus~\cite{Parrow2015}.
The intuitionistic modal logic $\FM$, in Sec.~\ref{section:modal}, complements those papers by characterising open barbed bisimilarity for the applied $\pi$-calculus.

Since concepts such as static equivalence are standard, many aspects of existing implementations of equivalence checkers can be reused. Future work includes adapting existing decision procedures~\cite{Tiu2010CSF,Cortier2017} to open barbed bisimilarity for the applied $\pi$-calculus.  We believe open barbed bisimilarity can be used to tame the problem with trace equivalence, where unnecessarily many constraints are generated~\cite{Cheval2017}. In related process models~\cite{Paige1987,Pistore1996} bisimilarity is demonstrably more efficient than trace equivalence. Open barbed bisimilarity complements existing work on partial-order reduction~\cite{Baelde2015} in this direction.

\section{Open bisimilarity: lazier approach; too fine for privacy properties}
\label{section:open}

Open bisimilarity for the applied $\pi$-calculus is referred to in the introduction and conclusion.
Open bisimilarity for applied $\pi$-calculus has not been defined in the literature, so we provide a short explanation.
This appendix clarifies why open bisimilarity is not the main subject of the text, and instead we focus on \textit{open barbed bisimilarity}.

By shifting to a late labelled transition system, we know how to define open bisimilarity for the applied $\pi$-calculus.
Previous definitions of open bisimilarity for a cryptographic calculus, the \textit{spi-calculus}, have complex definitions~\cite{briais06entcs,tiu07aplas}.
The approach of the current paper requires less machinery since the applied $\pi$-calculus is more abstract than the spi-calculus; hence abstracts away from details regarded as implementation concerns.
Open bisimilarity is defined in terms of a late labelled transition system, differing for the rules presented in Fig.~\ref{fig:late}.



\begin{figure}[h]
\[
\begin{gathered}
\begin{array}{c}
\begin{prooftree}
\justifies
\mathopen{\cin{M}{x}.}P \lts{M(x)} {P}
\using
\mbox{\textsc{Inp}}
\end{prooftree}
\qquad
\begin{prooftree}
 P \lts{\co{M}(x)} \nu \mathopen{\vec{z}.}\left(\sub{x}{N} \cpar P'\right)
\quad
 Q \lts{M(x)} Q' 
\quad
 \vec{z} \cap \fv{Q} = \emptyset
\justifies
 {P \cpar Q}\lts{\tau}{\nu \vec{z}.\left(P' \cpar Q'\sub{x}{N}\right)} 
\using
\mbox{\textsc{Close-l}}
\end{prooftree}
\\[15pt]
\begin{prooftree}
 P \lts{\co{M}(x)} \mathopen{\nu \vec{z}.}\left(\sub{x}{N} \cpar Q\right)
\qquad
 P \lts{M(x)} R
\qquad
 \vec{z} \cap \fv{P} = \emptyset
\justifies
 \bang P \lts{\tau} \mathopen{\nu \vec{z}.}\left( Q \cpar R\sub{x}{N}  \cpar \bang P \right)
\using
\mbox{\textsc{Rep-close}}
\end{prooftree}
\end{array}
\end{gathered}
\]
\caption{Rules of an open late labelled transition system, plus symmetric rules for parallel composition.
Note, for the fragment without mismatch, transition rules do not carry an environment.
}\label{fig:late}
\end{figure}
Instead of sets of private names, we employ histories representing the order in which messages are sent and received.
Respectful substitutions, defined over histories, are key to the lazy approach of open bisimilarity.
\begin{definition}[histories]\label{def:open-respects}
A \textit{history} is defined by grammar $h \Coloneqq \epsilon \mid h \cdot x^o \mid h \cdot M^i$.
Substitution $\sigma$ respects history $h$,
whenever for all histories $h'$ and $h''$
such that $h = h' \cdot x^o \cdot h''$, 
$x\sigma = x$,
and $y \in \fv{h'}$ implies $x \not\in y\sigma$.
\end{definition}

For clarity, we restrict to the fragment without mismatch. 
Hence the late labelled transitions in Fig.~\ref{fig:late} do not need to carry around environment information to resolve mismatches. 
This fragment is adequate for this discussion on related work, since open bisimilarity for the spi-calculus as only previously defined without mismatch.

\begin{definition}[open bisimilarity]\label{definition:open}
A symmetric relation indexed by a history $\mathrel{\mathcal{R}}$ is an open bisimulation whenever:
if $A \mathrel{\mathcal{R}}^{h } B$ the following hold, for $x$ fresh for $A$, $B$, $h$:
\begin{itemize}
\item $A$ and $B$ are statically equivalent.
\item Whenever $\sigma$ respects $h$, we have $A\sigma \mathrel{\mathcal{R}}^{h\sigma} B\sigma$.
\item If $A \lts{\tau} A'$ there exists $B'$ such that $ B \lts{\tau} B'$ and $A' \mathrel{\mathcal{R}}^{h } B'$.
\item If $A \lts{\co{M}(x)} A'$, for some $B'$, we have $ B \lts{\co{M}(x)} B'$
and $A' \mathrel{\mathcal{R}}^{h \cdot x^o} B'$.
\item If $A \lts{{M}(x)} A'$, for some $B'$, we have $B \lts{{M}(x)} B'$
and $A' \mathrel{\mathcal{R}}^{h \cdot x^i} B'$.
\end{itemize}
Open bisimilarity $\lsim_o$ is defined such that $P \mathrel{\lsim_o} Q$ holds whenever there exists open bisimulation $\mathcal{R}$
such that $P \mathrel\mathcal{R}^{x_1^i\cdot \hdots x_n^i} Q$ holds, where $\fv{P}\cup\fv{Q} \subseteq \left\{x_1, \hdots x_n\right\}$.
\end{definition}

Open bisimilarity is a congruence relation. The proof follows the same pattern as Theorem~3.6.
\begin{theorem}[congruence]
If $P \mathrel{\sim_o} Q$, then for all contexts $\context{\ \cdot\ }$, we have $\context{P} \mathrel{\sim_o} \context{Q}$.
\end{theorem}
The fact that open bisimilarity is preserved in all contexts is sufficient to show open bisimilarity is sound with respect to the greatest bisimilarity congruence, open barbed bisimilarity.
The proof is absolutely identical to Corollary~3.9.
\begin{corollary}[soundness]
If $P \mathrel{\sim_o} Q$ (open bisimilarity) then $P \simeq Q$ (open barbed bisimilarity).
\end{corollary}

We make two observations. Firstly, open bisimilarity is not adequate for certain privacy properties, such as the running example from the introduction.
Open bisimilarity can be extended to handle mismatch, by indexing open bisimulations and labelled transitions by both a history and a finite set of inequalities.
However, any conservative extension of open bisimilarity does not induce the law of excluded middle for guards involving messages that behave like private names; and hence can detect attacks that do not exist.
Secondly, the definition of quasi-open bisimilarity is undeniably simpler than Def.~\ref{definition:open} requiring only the keyword ``open'', compared to labelled bisimilarity.
Philosophically speaking, by Occam's razor, the simpler model is more likely the better choice.

\subsection{Open bisimilarity is conservative with respect to the spi-calculus.}

To strongly situate the current work with respect to notions of bisimilarity for cryptographic calculi in literature, we compare the applied $\pi$-calculus to the spi-calculus~\cite{abadi99ic}.
The spi-calculus is a more concrete predecessor of the applied $\pi$-calculus which is hard-wired with a fixed Dolev-Yao model for messages.

In the fixed message theory of the spi-calculus we assume we have pairs $\pair{M}{N}$ and symmetric encryption $\enc{M}{N}$.
Furthermore, to capture the expressive power of the spi-calculus using static equivalence in the applied $\pi$-calculus,
we also require the corresponding deconstructors $\dec{M}{N}$, $\fst{M}$ and $\snd{M}$
and equational theory $D$ such that $\fst{\pair{M}{N}} =_D M$, $\snd{\pair{M}{N}} =_D N$, and $\dec{\enc{M}{K}}{K} =_D M$.
In addition, in order for static equivalence to have the standard distinguishing power of the spi-calculus, which is type aware,
we require terms $\texttt{is\_enc}(N)$ and $\ok$ along with equation $\texttt{is\_enc}(\enc{M}{K}) =_D \ok$.

Spi-calculus processes, as with applied $\pi$-calculus processes, feature deadlock, input, output, parallel compositions, new name restriction, match and replication. 
Message terms however are only formed from variables and constructors $\pair{M}{N}$ and $\enc{M}{N}$.
Instead of deconstructors, explicit processes terms \texttt{case} and \texttt{let} are provided for decrypting messages and decomposing pairs.
The syntax of spi-calculus processes is defined as follows.
\begin{gather*}
\begin{array}{c}
x, y \quad\mbox{variables}
\\\\
\begin{array}{rlr}
M, N \Coloneqq & x \\
       \mid & \enc{M}{N} \\
       \mid & \pair{M}{N} 
\end{array}
\end{array}
\qquad\qquad
\begin{array}{rlr}
P \Coloneqq & 0 \\
       \mid & \cin{M}{x}.P \\
       \mid & \cout{M}{N}.P \\
       \mid & P \cpar P \\
       \mid & \nu x. P \\
       \mid & \match{M = N}P \\
       \mid & \bang P \\
       \mid & \ccase{M}{x}{N}{P} \\
       \mid & \clet{\pair{x}{y}}{M}{P} 
\end{array}
\end{gather*}
Spi-calculus processes are embedded directly as applied $\pi$-calculus processes by using the following mapping from \texttt{case} and \texttt{let} statements to applied $\pi$-calculus processes with deconstructors.
\[
\begin{array}{rl}
\embed{\ccase{M}{x}{K}{P}}~~~~=& \mathopen{\left[\enc{\dec{M}{K}}{K} = M \right]}\embed{P\sub{x}{\dec{M}{K}}}
\\
\embed{\clet{\pair{x}{y}}{M}{P}}~~~~=& \mathopen{\left[\pair{\fst{M}}{\snd{M}} = M \right]}\embed{P\sub{x,y}{\fst{M},\snd{M}}}
\end{array}
\]
Notice that the hard-wired theory for the spi-calculus permits successful decryption to be detected, using the guard $\enc{\dec{M}{K}}{K} = M$. Similarly, we can detect whether a message is a pair by using guard $\pair{\fst{M}}{\snd{M}} = M$. Thus the guards above only permit progress when decryption or projection, respectively, is successful.

By restricting to the above fragment of the applied $\pi$-calculus, we obtain the following result.
Open bisimilarity for the ``spi-calculus fragment'' of the applied $\pi$-calculus
coincides with open bisimilarity for the spi-calculus.
\begin{proposition}
If $P$ and $Q$ are spi-calculus processes then $P$ is open bisimilar to $Q$, as defined in related work~\cite{briais06entcs,tiu07aplas},
if and only if $\embed{P}$ is open bisimilar to $\embed{Q}$ according to Def.~\ref{definition:open}.  
\end{proposition}
The proof of the above proposition involves translating between the styles of the spi-calculus and applied $\pi$-calculus.
Open bisimilarity~\cite{briais06entcs,tiu07aplas} for the spi-calculus is defined using ``hedges''~\cite{borgstrom05mscs} --- a data structure representing indistinguishable pairs of messages exposed to the environment. 
For extended processes $\mathopen{\nu \vec{y}.}\left( \sigma \cpar P \right)$ and $\mathopen{\nu \vec{z}.}\left( \theta \cpar Q \right)$, the hedge is a list of pairs $[(x_1\sigma,x_1\theta), \hdots (x_n\sigma,x_n\theta)]$, where $\left\{ x_1,\hdots, x_n \right\} = \dom{\theta}$.

\subsection{Note on the implementation of open bisimilarity and quasi-open bisimilarity.}

The main purpose of the observation in this section is to situate the current work with respect to existing work on bisimilarity congruences for cryptographic calculi. In addition, this observation emphasises that established decision procedures for open bisimilarity developed for the spi-calculus~\cite{Tiu2010CSF,Tiu2016} lift to the setting of the applied $\pi$-calculus. An ongoing challenge for future work is to adapt decision procedures to further message theories. This is already an established research direction in cryptographic protocol analysis, since the problem can be reduced to deciding whether static equivalence holds, under all respectful substitutions.

Although open bisimilarity cannot verify the privacy of Server B from the introduction; it can be used to discover the attack on Server C, or even unlinkability attacks on e-passports. 
Since open bisimilarity is less expensive than quasi-open bisimilarity, we propose the following methodology. Firstly, search for an open bisimulation. If the search fails, construct an attack in intuitionistic $\FM$. If the construction fails, we have not found a real attack; hence continue to search for a quasi-open bisimulation. 
In this way, the more expensive quasi-open bisimilarity is only employed lazily, when necessary. Future work will evaluate the effectiveness of this implementation strategy.


\section{Conclusion}

This is the first thorough investigation into bisimilarities for the applied $\pi$-calculus that are congruences.
In cryptographic calculi in general, a bisimilarity congruence, \textit{open bisimilarity}, has previously been introduced for the spi-calculus~\cite{briais06entcs,tiu07aplas}.
However, work on the spi-calculus did not handle mismatch, and is less abstract, being hard-wired with a fixed message theory.
By moving to the coarser setting of \textit{open barbed bisimilarity} and lifting to the applied $\pi$-calculus, we are able to handle mismatch and any message theory,
in such a way that privacy-type properties can be verified.
Privacy properties addressed, elaborated on in Sec.~\ref{section:privacy}, involve \textit{if-then-else} with a guard depending on private information.
Equivalences significantly finer, such as diff-equivalence~\cite{Blanchet2008,Cheval2013}, are incomplete and hence may suggest attacks that do not exist.
Equivalences coarser than quasi-open bisimilarity are either not congruences or are not bisimilarities, by Theorem~\ref{theorem:openbarbed}.

Although nothing in cryptography is simple, definitions we introduce are concise and general. 
In order to define open barbed bisimilarity (Def.~\ref{def:open-barbed}), we require only the following ingredients:
\begin{itemize}
\item A notion of fresh substitution (Def.~\ref{def:respects}).
\item An open early labelled transition system (Fig.~\ref{figure:active}).
\end{itemize}
Open barbed bisimilarity is then defined using three succinct clauses.
Quasi-open bisimilarity (Def~\ref{definition:quasi-open}), the labelled alternative to open barbed bisimilarity,
is also concise.
For quasi-open bisimilarity the additional device required is the standard definition of static equivalence (Def.~\ref{def:static}).
Open barbed bisimilarity provides an objective reference point --- any (strong) bisimilarity congruence must be sound with respect to open barbed bisimilarity. 
The main result of this paper, Theorem~\ref{theorem:openbarbed}, verifies quasi-open bisimilarity coincides with the more objective language-independent open barbed bisimilarity.
Such an objective reference allows design decisions to be resolved, such as how to handle expressive message theories and \texttt{if-then-else}.

In terms of definitions, the gap between open barbed bisimilarity and the ``classical'' \textit{observational equivalence} is small --- ensure the relation is preserved in all contexts, not just contexts that introduce a new process in parallel.
The gap between quasi-open bisimilarity and ``classical'' \textit{labelled bisimilarity} is smaller still --- ensure the relation is preserved under reachability (Def.~\ref{def:close}).
However, the gap is significant, since in this work we obtain a congruence relation. 
Furthermore, a recent breakthrough~\cite{Ahn2017} provided us with the insight to logically characterise open barbed bisimilarity.
The insight is that, closing a suitable modal logic under reachability, we obtain a characteristic intuitionistic modal logic (Theorems~\ref{theorem:sound-modal}~and~\ref{theorem:complete-modal}).
This we believe is the first logical characterisation of any bisimilarity congruence for any cryptographic calculus.
Characteristic formulae can be used, for example, to describe privacy attacks whenever two processes are distinguished.
Note, the intuitionistic modal logic $\FM$ is likely to have applications beyond describing attacks on privacy.

For the ``classical'' labelled bisimilarity~\cite{Abadi16}, there will always be the following hanging question.
\begin{quote}
Can I reason compositionally, proving sub-protocols are correct, with the reassurance that correctness will still hold in a larger context?
\end{quote}
For labelled bisimilarity, the answer to the above question is not immediate --- it depends on the processes and the context.
An example of such a potential pitfall is explained in Section~\ref{section:compositionality}.
In contrast, open barbed bisimilarity removes the need to ask the above question. 
Any property verified using open barbed bisimilarity can be reused anywhere in another proof.

\begin{acks}
I thank Alwen Tiu and Ki Yung Ahn for their collaboration in project MOE2014-T2-2-076 (Singapore MOE Tier 2 grant), under which this theory was developed.
Our joint work on mismatch, presented at LICS'18, was a prerequisite for extending the theory presented to the full applied $\pi$-calculus.
\end{acks}

\bibliography{biblio}



\end{document}